\documentclass[conference]{IEEEtran}

\usepackage[letterpaper, left=0.7in, right=0.7in, bottom=1.1in, top=0.75in]{geometry}
\usepackage{url}
\usepackage[utf8]{inputenc}
\usepackage{xcolor}
\usepackage{amsmath}
\DeclareMathOperator*{\argmax}{arg\,max}
\DeclareMathOperator*{\argmin}{arg\,min}
\usepackage{multirow}
\usepackage{amssymb}
\usepackage{cite}
\usepackage{enumitem}
\usepackage{bm}
\usepackage{graphicx}
\usepackage{amsthm}
\newtheorem{thm}{Theorem}
\newtheorem{lem}[thm]{Lemma}

\usepackage{gensymb}

\usepackage[acronyms,nonumberlist,nopostdot,nomain,nogroupskip]{glossaries}
\usepackage{hyperref}
\usepackage{tablefootnote}
\usepackage{booktabs}
\usepackage{tabularx}
\usepackage{epsfig}
\usepackage{tikz}
\usetikzlibrary{external}
\usepackage{pgfplots}
\pgfplotsset{compat=newest} 
\pgfplotsset{plot coordinates/math parser=false}

\usetikzlibrary{plotmarks,patterns,patterns.meta,decorations.pathreplacing,backgrounds,calc,arrows,arrows.meta,spy,matrix}
\usepgfplotslibrary{patchplots,groupplots}
\usepackage{tikzscale}
\usepackage{siunitx}
\usepackage{tablefootnote}

\usepackage{multirow}
\usepackage{algorithm2e, setspace}
\SetAlCapNameFnt{\footnotesize}
\SetAlCapFnt{\footnotesize}
\RestyleAlgo{ruled}
\SetKwComment{Comment}{/* }{ */}
\usepackage{algpseudocode}

\usepackage[font=scriptsize]{subcaption}
\usepackage[font=footnotesize]{caption}

\usepackage{mathtools}

\usepackage{dblfloatfix}    % To enable figures at the bottom of page
\usepackage{colortbl}
\usepackage{flushend}
\usepackage{footnote}

\usepackage{lipsum,afterpage,refcount}

\newacronym{3gpp}{3GPP}{3rd Generation Partnership Project}
\newacronym{adc}{ADC}{Analog to Digital Converter}
\newacronym{5g}{5G}{5th generation}
\newacronym{6g}{6G}{6th generation}
\newacronym{ai}{AI}{Artificial Intelligence}
\newacronym{aimd}{AIMD}{Additive Increase Multiplicative Decrease}
\newacronym{am}{AM}{Acknowledged Mode}
\newacronym{amc}{AMC}{Adaptive Modulation and Coding}
\newacronym{aqm}{AQM}{Active Queue Management}
\newacronym{awgn}{AGWN}{Additive White Gaussian Noise}
\newacronym{balia}{BALIA}{Balanced Link Adaptation}
\newacronym{bdp}{BDP}{Bandwidth-Delay Product}
\newacronym{bf}{BF}{beamforming}
\newacronym{cc}{CC}{Congestion Control}
\newacronym{cdf}{CDF}{Cumulative Distribution Function}
\newacronym{cn}{CN}{Core Network}
\newacronym{cqi}{CQI}{Channel Quality Information}
\newacronym{cp}{CP}{Control Plane}
\newacronym{csirs}{CSI-RS}{Channel State Information - Reference Signal}
\newacronym{dc}{DC}{Dual Connectivity}
\newacronym{rb}{RB}{Resource Block}
\newacronym{dce}{DCE}{Direct Code Execution}
\newacronym{dci}{DCI}{Downlink Control Information}
\newacronym{udp}{UDP}{User Datagram Protocol}
\newacronym{dl}{DL}{downlink}
\newacronym{fcfs}{FCFS}{first-come-first-served}
\newacronym{dmr}{DMR}{Deadline Miss Ratio}
\newacronym{fspl}{FSPL}{free-space path loss}
\newacronym{dmrs}{DMRS}{DeModulation Reference Signal}
\newacronym{e2e}{E2E}{End-to-End}
\newacronym{ppp}{PPP}{Poission Point Process}
\newacronym{aoi}{AoI}{Area of Interest}
\newacronym{cpu}{CPU}{Central Processing Unit}
 \newacronym{gpu}{GPU}{Graphics Processing Unit}
 \newacronym{tpu}{TPU}{Tensor Processing Unit}
\newacronym{si}{SI}{Study Item}
\newacronym{ecn}{ECN}{Explicit Congestion Notification}
\newacronym{edf}{EDF}{Earliest Deadline First}
\newacronym{enb}{eNB}{eNodeB}
\newacronym{epc}{EPC}{Evolved Packet Core}
\newacronym{es}{ES}{Edge Server}
\newacronym{cav}{CAV}{Connected and Autonomous Vehicle}
\newacronym{fdma}{FDMA}{Frequency Division Multiple Access}
\newacronym{fdd}{FDD}{Frequency Division Duplexing}
\newacronym{upa}{UPA}{Uniform Planar Array}
\newacronym{car}{CAR}{Circular Aperture Reflector }
\newacronym[firstplural=Radio Access Technologies (RATs)]{rat}{RAT}{Radio Access Technology}
\newacronym[firstplural=Radio Access Technology (RTs)]{rt}{RT}{Radio Technology}
\newacronym{fs}{FS}{Fast Switching}
\newacronym{isd}{ISD}{inter-site distance}
\newacronym{ftp}{FTP}{File Transfer Protocol}
\newacronym{gnb}{gNB}{Next Generation Node Base}
\newacronym{harq}{HARQ}{Hybrid Automatic Repeat reQuest}
\newacronym{hetnet}{HetNet}{Heterogeneous Network}
\newacronym{hh}{HH}{Hard Handover}
\newacronym{hol}{HOL}{Head-of-Line}
\newacronym{ia}{IA}{Initial Access}
\newacronym{imt}{IMT}{International Mobile Telecommunication}
\newacronym{iot}{IoT}{Internet of Things}
\newacronym{los}{LOS}{Line of Sight}
\newacronym{lte}{LTE}{Long Term Evolution}
\newacronym{m2m}{M2M}{Machine to Machine}
\newacronym{mac}{MAC}{Medium Access Control}
\newacronym{mc}{MC}{Multi-Connectivity}
\newacronym{mcs}{MCS}{Modulation and Coding Scheme}
\newacronym{mec}{MEC}{Mobile Edge Cloud}
\newacronym{mi}{MI}{Mutual Information}
\newacronym{mimo}{MIMO}{Multiple Input Multiple Output}
\newacronym{mmwave}{mmWave}{millimeter wave}
\newacronym{mptcp}{MPTCP}{Multipath TCP}
\newacronym{mr}{MR}{Maximum Rate}
\newacronym{mss}{MSS}{Maximum Segment Size}
\newacronym{mtd}{MTD}{Machine-Type Device}
\newacronym{mtu}{MTU}{Maximum Transmission Unit}
\newacronym{nfv}{NFV}{Network Function Virtualization}
\newacronym{vnf}{VNF}{Virtualization Network Function}
\newacronym{gv}{GV}{ground vehicle}
\newacronym{gvs}{GVs}{ground vehicles}
\newacronym{vec}{VEC}{Vehicular Edge Computing}
\newacronym{sdn}{SDN}{Software Defined Networking}
\newacronym{nlos}{NLOS}{Non Line of Sight}
\newacronym{nlosb}{NLOSb}{Building Non Line of Sight}
\newacronym{nlosv}{NLOSv}{Vehicle Non Line of Sight}
\newacronym{nr}{NR}{New Radio}
\newacronym{ofdm}{OFDM}{Orthogonal Frequency Division Multiplexing}
\newacronym{pdcch}{PDCCH}{Physical Downlink Control Channel}
\newacronym{pdcp}{PDCP}{Packet Data Convergence Protocol}
\newacronym{pdsch}{PDSCH}{Physical Downlink Shared Channel}
\newacronym{pdu}{PDU}{Packet Data Unit}
\newacronym{pf}{PF}{Proportional Fair}
\newacronym{pgw}{PGW}{Packet Gateway}
\newacronym{phy}{PHY}{Physical}
\newacronym{pbch}{PBCH}{Physical Broadcast Channel}
\newacronym[plural=\gls{mme}s,firstplural=Mobility Management Entities (MMEs)]{mme}{MME}{Mobility Management Entity}
\newacronym{prb}{PRB}{Physical Resource Block}
\newacronym{pss}{PSS}{Primary Synchronization Signal}
\newacronym{pucch}{PUCCH}{Physical Uplink Control Channel}
\newacronym{pusch}{PUSCH}{Physical Uplink Shared Channel}
\newacronym{rach}{RACH}{Random Access Channel}
\newacronym{ran}{RAN}{Radio Access Network}
\newacronym{red}{RED}{Random Early Detection}
\newacronym{rf}{RF}{Radio Frequency}
\newacronym{rlc}{RLC}{Radio Link Control}
\newacronym{rlf}{RLF}{Radio Link Failure}
\newacronym{rrc}{RRC}{Radio Resource Control}
\newacronym{rrm}{RRM}{Radio Resource Management}
\newacronym{rr}{RR}{Round Robin}
\newacronym{rs}{RS}{Remote Server}
\newacronym{rsrp}{RSRP}{Reference Signal Received Power}
\newacronym{rss}{RSS}{Received Signal Strength}
\newacronym{rtt}{RTT}{Round Trip Time}
\newacronym{rw}{RW}{Receive Window}
\newacronym{rx}{RX}{Receiver}
\newacronym{sa}{SA}{standalone}
\newacronym{sack}{SACK}{Selective Acknowledgment}
\newacronym{sap}{SAP}{Service Access Point}
\newacronym{sch}{SCH}{Secondary Cell Handover}
\newacronym{scoot}{SCOOT}{Split Cycle Offset Optimization Technique}
\newacronym{sdma}{SDMA}{Spatial Division Multiple Access}
\newacronym{sinr}{SINR}{Signal to Interference plus Noise Ratio}
\newacronym{sm}{SM}{Saturation Mode}
\newacronym{snr}{SNR}{Signal-to-Noise Ratio}
\newacronym{son}{SON}{Self-Organizing Network}
\newacronym{ss}{SS}{Synchronization Signal}
\newacronym{srs}{SRS}{Sounding Reference Signal}
\newacronym{sss}{SSS}{Secondary Synchronization Signal}
\newacronym{tb}{TB}{Transport Block}
\newacronym{tcp}{TCP}{Transmission Control Protocol}
\newacronym{tdd}{TDD}{Time Division Duplexing}
\newacronym{tdma}{TDMA}{Time Division Multiple Access}
\newacronym{tfl}{TfL}{Transport for London}
\newacronym{tm}{TM}{Transparent Mode}
\newacronym{prr}{PRR}{Packet Reception Ratio}
\newacronym{trp}{TRP}{Transmitter Receiver Pair}
\newacronym{tti}{TTI}{Transmission Time Interval}
\newacronym{ttt}{TTT}{Time-to-Trigger}
\newacronym{tx}{TX}{Transmitter}
\newacronym{ue}{UE}{User Equipment}
\newacronym{ul}{UL}{uplink}
\newacronym{uml}{UML}{Unified Modeling Language}
\newacronym{um}{UM}{Unacknowledged Mode}
\newacronym{utc}{UTC}{Urban Traffic Control}
\newacronym{vm}{VM}{Virtual Machine}
\newacronym{rsrq}{RSRQ}{Reference Signal Received Quality}
\newacronym{rssi}{RSSI}{Received Signal Strength Indicator}
\newacronym{crs}{CRS}{Cell Reference Signal}
\newacronym{v2v}{V2V}{Vehicle-to-Vehicle}
\newacronym{v2i}{V2I}{Vehicle-to-Infrastructure}
\newacronym{v2n}{V2N}{Vehicle-to-Network}
\newacronym{v2x}{V2X}{Vehicle-to-Everything}
\newacronym{vn}{VN}{Vehicular Node}
\newacronym{dsrc}{DSRC}{Dedicated Short Range Communication}
\newacronym{ci}{CI}{context information}
\newacronym{voi}{VoI}{value of information}
\newacronym{gps}{GPS}{Global Positioning System}
\newacronym{qos}{QoS}{Quality of Service}
\newacronym{qoe}{QoE}{Quality of Experience}
\newacronym{ml}{ML}{Machine Learning}
\newacronym{ahp}{AHP}{Analytic Hierarchy Process}
\newacronym{lidar}{LIDAR}{Light Detection and Ranging}
\newacronym{sumo}{SUMO}{Simulation of Urban MObility}
\newacronym{wave}{WAVE}{Wireless Access in Vehicular Environment}
\newacronym{c-its}{C-ITS}{Connected Intelligent Transportation System}
\newacronym{dash}{DASH}{Dynamic Adaptive Streaming over HTTP}
\newacronym{http}{HTTP}{HyperText Transfer Protocol}
\newacronym{nt}{NT}{Non-Terrestrial}
\newacronym{ntc}{NTC}{non-terrestrial communication}
\newacronym{ntn}{NTN}{non-terrestrial network}
\newacronym{tn}{TN}{terrestrial network}
\newacronym{hap}{HAP}{High Altitude Platform}
\newacronym{leo}{LEO}{Low Earth Orbit}
\newacronym{meo}{MEO}{Medium Earth Orbit}
\newacronym{geo}{GEO}{Geostationary Earth Orbit}
\newacronym{uav}{UAV}{Unmanned Aerial Vehicle}
\newacronym{nsat}{nSAT}{Nanosatellite}
\newacronym{ehf}{EHF}{extremely high-frequency}
\newacronym{ioe}{IoE}{Internet of Everyone}
\newacronym{gan}{GaN}{Gallium Nitride}
\newacronym{tle}{TLE}{two-line element}
\newacronym{ecdf}{ECDF}{Empirical Cumulative Distribution Function}
\newacronym{fifo}{FIFO}{First-Input First-Output}
\newacronym{gnss}{GNSS}{Global Navigation Satellite System}
\newacronym{essa}{ESSA}{Enhanced Synchronized Slot Allocation}
\newacronym{ta}{TA}{Timing Advance}

\usepackage{tikz}
\usepackage{pgfplots}
\usepackage{tikzscale}
\usepackage{tikz-qtree}

\pgfplotsset{compat=newest}
\pgfplotsset{plot coordinates/math parser=false}
\pgfplotsset{every axis/.append style={
                    label style={font=\scriptsize},
                    tick label style={font=\scriptsize},
                    legend style={font=\scriptsize}
                    }}
\usetikzlibrary{plotmarks,shapes,patterns,decorations.pathreplacing,backgrounds,calc,arrows,arrows.meta,spy,matrix,shadows,trees,positioning,fit}
\usepgfplotslibrary{patchplots,groupplots}

\tikzstyle{startstop} = [rectangle, rounded corners, minimum width=2cm, minimum height=0.5cm,text centered, draw=black]
\tikzstyle{io} = [trapezium, trapezium left angle=70, trapezium right angle=110, minimum width=3cm, minimum height=1cm, text centered, draw=black]
\tikzstyle{process} = [rectangle, minimum width=2cm, minimum height=0.5cm, text centered, draw=black, alignb=center]
\tikzstyle{decision} = [ellipse, minimum width=2cm, minimum height=1cm, text centered, draw=black]
\tikzstyle{arrow} = [thick,<->,>=stealth]
\tikzstyle{line} = [thick,>=stealth]
\tikzstyle{darrow} = [thick,<->,>=stealth,dashed]
\tikzstyle{sarrow} = [thick,->,>=stealth]
\tikzstyle{larrow} = [line width=0.1mm,dashdotted,->,>=stealth]

\pgfkeys{/pgf/number format/.cd,1000 sep={\,}} % thousand separator: space

\makeatletter
\def\grd@save@target#1{%
  \def\grd@target{#1}}
\def\grd@save@start#1{%
  \def\grd@start{#1}}
\tikzset{
  grid with coordinates/.style={
    to path={%
      \pgfextra{%
        \edef\grd@@target{(\tikztotarget)}%
        \tikz@scan@one@point\grd@save@target\grd@@target\relax
        \edef\grd@@start{(\tikztostart)}%
        \tikz@scan@one@point\grd@save@start\grd@@start\relax
        \draw[minor help lines] (\tikztostart) grid (\tikztotarget);
        \draw[major help lines] (\tikztostart) grid (\tikztotarget);
        \grd@start
        \pgfmathsetmacro{\grd@xa}{\the\pgf@x/1cm}
        \pgfmathsetmacro{\grd@ya}{\the\pgf@y/1cm}
        \grd@target
        \pgfmathsetmacro{\grd@xb}{\the\pgf@x/1cm}
        \pgfmathsetmacro{\grd@yb}{\the\pgf@y/1cm}
        \pgfmathsetmacro{\grd@xc}{\grd@xa + \pgfkeysvalueof{/tikz/grid with coordinates/major step x}}
        \pgfmathsetmacro{\grd@yc}{\grd@ya + \pgfkeysvalueof{/tikz/grid with coordinates/major step y}}
        \foreach \x in {\grd@xa,\grd@xc,...,\grd@xb}
        \node[anchor=north] at (\x,\grd@ya) {\pgfmathprintnumber{\x}};
        \foreach \y in {\grd@ya,\grd@yc,...,\grd@yb}
        \node[anchor=east] at (\grd@xa,\y) {\pgfmathprintnumber{\y}};
      }
    }
  },
  minor help lines/.style={
    help lines,
    gray,
    line cap =round,
    xstep=\pgfkeysvalueof{/tikz/grid with coordinates/minor step x},
    ystep=\pgfkeysvalueof{/tikz/grid with coordinates/minor step y}
  },
  major help lines/.style={
    help lines,
    line cap =round,
    line width=\pgfkeysvalueof{/tikz/grid with coordinates/major line width},
    xstep=\pgfkeysvalueof{/tikz/grid with coordinates/major step x},
    ystep=\pgfkeysvalueof{/tikz/grid with coordinates/major step y}
  },
  grid with coordinates/.cd,
  minor step x/.initial=.5,
  minor step y/.initial=.2,
  major step x/.initial=1,
  major step y/.initial=1,
  major line width/.initial=1pt,
}
\makeatother

\newlength\fheight
\newlength\fwidth

\definecolor{steelblue}{RGB}{176,196,222}

\usepackage[capitalize]{cleveref}
\crefname{section}{Sec.}{Secs.}
\usetikzlibrary{decorations}

\IEEEoverridecommandlockouts
\newcommand\copyrightnotice{%
\begin{tikzpicture}[remember picture,overlay]
\node[anchor=south,yshift=15pt] at (current page.south) {\fbox{\parbox{\dimexpr\textwidth-\fboxsep-\fboxrule\relax}{
\footnotesize \textcopyright 2024 IEEE. Personal use of this material is permitted.
Permission from IEEE must be obtained for all other uses, in any current or future media,
including reprinting/republishing this material for advertising or promotional purposes,
creating new collective works, for resale or redistribution to servers or lists,
or reuse of any copyrighted component of this work in other works.}}};
\end{tikzpicture}
}

\begin{document}
\bstctlcite{IEEEexample:BSTcontrol}

\title{Enhanced Time Division Duplexing Slot Allocation and Scheduling in Non-Terrestrial Networks}

    \author{\IEEEauthorblockN{Alessandro Traspadini, Marco Giordani, Michele Zorzi \medskip}
%\IEEEauthorblockA{University of Padova, Italy. Email: \texttt{\{bonora, traspadini, giordani, zorzi\}@dei.unipd.it}\\
%}}

\IEEEauthorblockA{ Department of Information Engineering, University of Padova, Italy.\\
Email:	\texttt{\{traspadini, giordani, zorzi\}@dei.unipd.it}}}

\maketitle

\copyrightnotice

\begin{abstract}
The integration of \glspl{ntn} and \glspl{tn} is fundamental for extending connectivity to rural and underserved areas that lack coverage from traditional cellular infrastructure.
However, this integration presents several challenges.  %related to the different latency, propagation, and hardware constraints of \gls{tn} and \gls{ntn}.
For instance, \glspl{tn} mainly operate in \gls{tdd}. However, for NTN via satellites, \gls{tdd} is complicated due to synchronization problems in large cells, and the significant impact of guard periods and long propagation delays.
In this paper, we propose a novel slot allocation mechanism to enable \gls{tdd} in \gls{ntn}.
This approach permits to allocate additional transmissions during the guard period between a downlink slot and the corresponding uplink slot to reduce the overhead, provided that they do not interfere with other concurrent transmissions.
Moreover, we propose two scheduling methods to select the users that transmit based on considerations related to the \gls{snr} or the propagation delay.
Simulations demonstrate that our proposal can increase the network capacity compared to a benchmark scheme that does not schedule transmissions in guard~periods.
\end{abstract}

\glsresetall


\begin{tikzpicture}[remember picture,overlay]
\node[anchor=north,yshift=-10pt] at (current page.north) {\parbox{\dimexpr\textwidth-\fboxsep-\fboxrule\relax}{
\centering\footnotesize This paper has been accepted for publication at the 58th Asilomar Conference on Signals, Systems, and Computers. \textcopyright 2024 IEEE.\\
Please cite it as: A. Traspadini, M. Giordani, and M. Zorzi, "Enhanced Time Division Duplexing Slot Allocation and Scheduling in Non-Terrestrial Networks," 58th Asilomar Conference on Signals, Systems, and Computers, 2024.}};
\end{tikzpicture}

\glsresetall

\section{Introduction}
\label{sec:intro}
The research community is studying \glspl{ntn} as an approach to extend ground coverage in extreme environments, such as in remote or unconnected areas~\cite{Chaoub20216g}, and provide connectivity in case of emergency~\cite{Emerging24Shahid}.
%Indeed, the coordination of satellite communications and terrestrial mobile networks will ensure the availability, continuity, ubiquity and scalability of the service~\cite{Toward22Araniti}.
Notably, the integration and coordination between \glspl{ntn} and \glspl{tn} is particularly promising to improve service continuity, availability, and scalability~\cite{Toward22Araniti}.

In this context, the \gls{3gpp} has consolidated the possible use of satellites, especially \gls{leo} satellites, into the 5G \gls{nr} cellular standard since Release 15~\cite{38811,38821}, and a new  Work Item on 3GPP NTN has been officially approved in Release 17~\cite{3GPP_WINTN}. Direct-to-handset satellite communication for unmodified smartphones is also being promoted. 
However, to date, the 3GPP NTN standard has not yet been fully formalized, %and it is still unclear how 3GPP NTN protocols should be designed and implemented, 
also with respect to the baseline 5G NR terrestrial protocol stack.
The key challenges include the longer propagation delay and more severe propagation loss of the satellite channel with respect to \glspl{tn}, and the large Doppler shifts and velocity of satellite platforms~\cite{giordani2021non,IntroElJaafari23}. 
Several studies have proposed solutions to address these problems.
For instance, Han \emph{et al.} proposed to relocate the functions of mobility management to satellites to improve flexibility and reduce delays~\cite{Novel21Han}.
Similarly, Peng~\emph{et al.} investigated interference mitigation techniques to facilitate the integration of multi-beam satellite  systems with \glspl{tn}~\cite{Integ22Peng}.
The analysis on co-channel interference in integrated TN/NTN systems was further discussed in~\cite{Reverse21Lee,Interference24Lee}.
Additionally, Zu \emph{et al.} presented a cooperative transmission design for TN/NTN systems with the objective of improving coverage and capacity~\cite{Cooperative19Zhu}.

To date, a major compatibility issue between TNs and NTNs is related to the duplex communication. In fact, most satellite networks are designed in \gls{fdd}~\cite{3GPP_38863}. 
In \gls{fdd}, \glspl{ue} and satellites operate on different frequency bands for \gls{ul} and \gls{dl} slots. This approach permits full-duplex communication via simultaneous and continuous \gls{ul} and \gls{dl} transmissions, eliminates the need for tight synchronization, and promotes interference mitigation as \glspl{ue} communicate in orthogonal frequency bands.
In turn, terrestrial 5G NR networks operate in \gls{tdd}~\cite{38300}. The advantages include channel reciprocity, dynamic and flexible traffic allocation, lower hardware costs, and frequency diversity~\cite{EvolutionChan06}.
However, to guarantee full and seamleass integration between TNs and NTNs, the 3GPP suggests that NTNs be designed also in TDD~\cite{38811}, which may raise several concerns.
First, the length of guard periods, which ensure that consecutive \gls{ul} and \gls{dl} transmissions do not interfere, must be proportional to the cell size and propagation delay. Both are particularly large in satellite networks, and may span several transmission slots, with negative implications in terms of overhead.
Second, TDD requires synchronization, which is difficult in satellite networks due to differential delays that could be experienced by two or more \glspl{ue} within the same cell, especially at low elevation.
In 3GPP, \gls{ta} was introduced to compensate for the propagation delays between UEs and the satellite to achieve synchronization. Still, this approach prevents the network from scheduling TDD transmissions during guard periods to avoid interference, which results in large overheads and inefficient use of resources.

To address these issues, in this paper we propose a novel slot allocation mechanism, called \gls{essa}, to improve capacity in \gls{tdd} \gls{ntn}. 
Specifically, the idea is to reduce the number of guard periods, where transmissions may not be formally scheduled, by allocating multiple (rather than just one) \gls{dl} slots before a \gls{ul} slot, if these transmissions do not generate interference at the satellite.  
Moreover, we propose two scheduling methods to select the optimal UEs that should transmit in TDD to maximize the network capacity. Scheduling is based on either the best link quality (MG) or the minimum differential delay (MS).
We evaluate the performance of ESSA, as well as of MG and MS, against a TDD benchmark scheme that does not schedule additional transmissions during guard periods, but only implements \gls{ta} for synchronization. We demonstrate via simulations that ESSA can increase the channel usage by nearly 50\%, and that the capacity of MS combined with ESSA is around three times higher than that of MG with TA.

The remainder of this paper is structured as follows.
In \cref{sec:system_model} we present the system model, in~\cref{sec:ssa} we introduce ESSA and the proposed scheduling methods, in~\cref{sec:simulation_results} we describe our simulation results, and in~\cref{sec:conclusions} we conclude the paper with suggestions for future work.

\begin{figure}[t!]
    \centering
	\includegraphics[width=0.99\columnwidth]{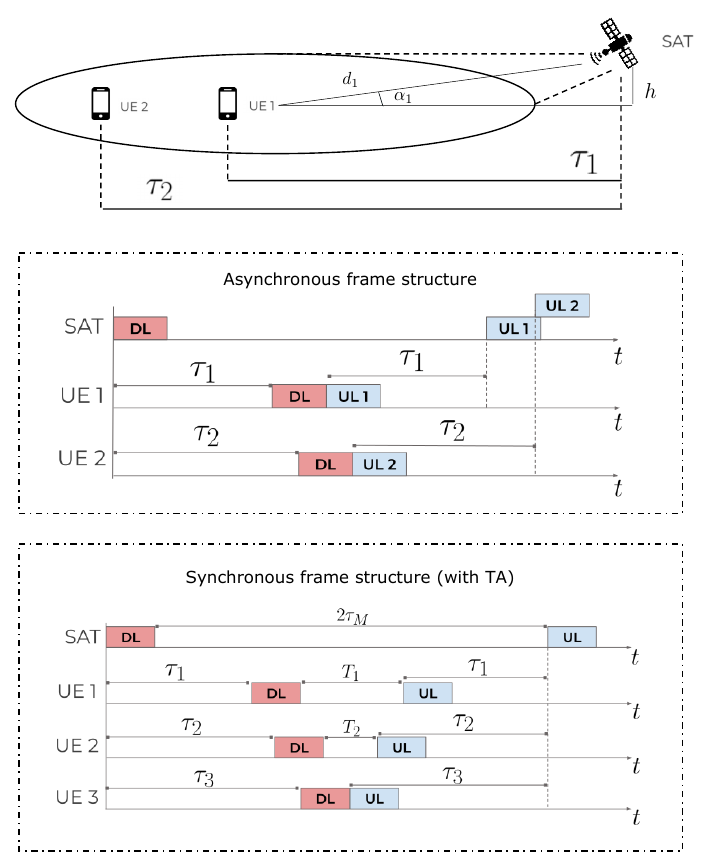}
 \caption{ Illustration of the scenario for $N_{\rm UE}=2$ (top), and the TDD frame structure with (middle) and without (bottom) timing advance.\vspace{-0.5cm}}
 \label{fig:fig1}
 \end{figure}

\section{System Model}
\label{sec:system_model}
Our NTN scenario, depicted in~\cref{fig:fig1} (top), consists of $N_{\rm UE}$ \glspl{ue} that are connected through direct-to-satellite links to a single \gls{leo} satellite (SAT) providing \gls{gnb} functionalities at an altitude $h$.
The system is assumed to be steady, with fixed propagation delays for both \gls{ul} and \gls{dl} transmissions, as the mobility of the satellite on its orbit and that of the \glspl{ue} can be neglected at the subframe level. 
As described in~\cite{Energy23Traspadini}, the distance $d_i$ from the satellite to a generic \gls{ue} $i$ is given by
\begin{equation}
d_{i} = \sqrt{R^2_E \sin^2(\alpha_i) + h^2 + 2 h R_E} - R_E \sin(\alpha_i),
\label{eq:distance}
\end{equation}
where $R_E$ is the Earth's radius, and $\alpha_i$ is the elevation angle between \gls{ue} $i$ and the satellite.
In the following, we describe the channel model (\cref{sub:channel}) and the TDD frame structure (\cref{sub:frame}).

\subsection{Channel Model}
\label{sub:channel}
In this study we consider the \gls{3gpp} \gls{ntn} channel model with shadowing in an urban scenario~\cite{38811}.\footnote{A complete characterization of the \gls{3gpp} \gls{ntn} channel model in~\cite{38811} is described in~\cite{sandri2023implementation}, and implemented in ns-3 here: \url{https://gitlab.com/mattiasandri/ns-3-ntn/-/tree/ntn-dev}.}
The path loss between the satellite and \gls{ue} $i$ is given by:
\begin{equation}
PL_{i} = FPL_{i} + A_{g} + A_{s} + SF
\end{equation}
where $FPL$ is the \gls{fspl}, $ A_{g}$ is the atmospheric absorption loss due to dry air and water vapor attenuation, and $A_{s}$ is the scintillation loss due to sudden changes in the refractive index caused by variations of the temperature, water vapor content, and
barometric pressure~\cite{wang2020potential}.
$SF$ represents the shadowing component, which is modeled as a zero mean log-normal random variable with variance $\sigma_s$.
The value of $\sigma_s$ depends on various factors including channel conditions, the scenario, and the elevation angle, as described in~\cite{38811}.
The \gls{fspl} is given by:
\begin{equation}
FPL_{i} = 92.45 + 20\log_{10}(f) + 20\log_{10}(d_i),
\end{equation}
where $f$ represents the carrier frequency in~GHz and $d_i$ is the distance in kilometers.
The received power $P^{rx}_{i}$ of \gls{ue} $i$ is expressed as
\begin{equation}
    P^{rx}_{i} = P^{tx} + G^{tx} + G^{rx} - PL_{i}, 
\end{equation}
where $P^{tx}$ is the transmitted power, $G^{tx}$ is the transmitter antenna gain, and $G^{rx}$ is the receiver antenna gain.
Thus, the \gls{snr} for \gls{ue} $i$ is given by
\begin{equation}
    \gamma_{i} =  P^{rx}_{i} - 10\log_{10}(kTB) - N_f,
\end{equation}
where $k$ is the Boltzmann constant, $T$ is the antenna noise temperature, $B$ is the bandwidth, and $N_f$ is the noise figure.
From the \gls{snr}, the Ergodic capacity $C_i$ of \gls{ue} $i$ is defined as
\begin{equation}
    C_{i} = B \log_{2} \left(1+10^{\gamma_i/10} \right).
    \label{eq:C}
\end{equation}

\begin{figure*}[t!]
\centering 
\includegraphics[width=0.9\textwidth]{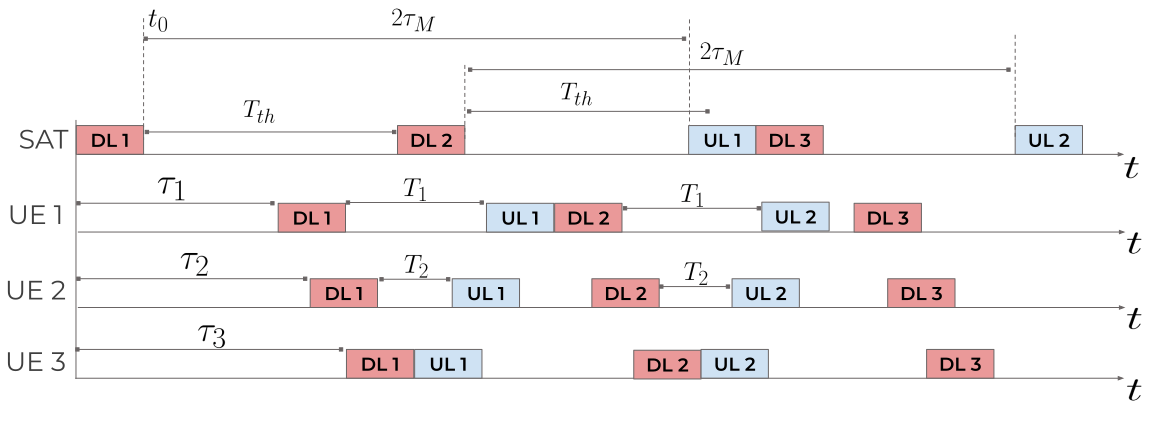}
\caption{\vspace{-0.5cm}Illustration of the proposed \acrfull{essa} mechanism.}
\label{fig:proposed}
\end{figure*}

\subsection{TDD Frame Structure}
\label{sub:frame}
The one-way propagation delay for \gls{ue} $i$ is given by
\begin{equation}
\tau_{i} = {d_i}/{c},
\end{equation}
where $c$ is the speed of light.
In~\cref{fig:fig1} (middle) we illustrate the slot allocation in TDD for two \glspl{ue} under the coverage of the same satellite gNB.
The satellite transmits a \gls{dl} slot, which is received by \gls{ue} 1 after $\tau_1$, and by \gls{ue} 2 after $\tau_2\gg\tau_1$.\footnote{For a cell radius of 1000 km, the differential propagation delay of UEs in the cell can be up to 10 ms~\cite[Table 7.2.1.1.1.2-1]{38821}.}
If a \gls{ul} slot for both UEs is allocated immediately after the corresponding \gls{dl} slot, synchronization of the \gls{ul} slots at the satellite reference point is not achieved due to the different propagation delays of the \glspl{ue}.
This misalignment results in uncoordinated transmissions, causing interference.

To address this issue, a \gls{ta} mechanism should be applied~\cite{38213,38133}.
In \gls{3gpp} Release 17, each \gls{ue} is assumed to have knowledge of both its own position and the satellite position using \gls{gnss} coordinates, which permits to estimate the propagation delay of the UL channel~\cite{NTN24Saad}.
\Glspl{ue} can then report this estimate to the satellite, which computes the longest propagation delay in the cell, denoted as $\tau_M$, corresponding to the furthest \gls{ue} within the coverage area of the satellite.
This information is included in the \gls{pdcch} to perform \gls{ta} and maintain synchronization.
As shown in~\cref{fig:fig1} (bottom), 
%if the satellite schedules an \gls{ul} slot after the reception of the \gls{dl} slot, 
if each \gls{ue} $i$ applies TA of duration $T_i = 2(\tau_M - \tau_i)$, i.e., UL slots are ``advanced'' by $T_i$ relative to the corresponding \gls{dl} slots, then \gls{ul} slots are also synchronized at the satellite reference point.
In addition, the satellite can schedule different \gls{ofdm} symbols within the same \gls{ul} slot to different \glspl{ue}, without causing interference.
 
However, in this setup, the guard period between a \gls{dl} and an \gls{ul} slot must be equal to $2\tau_M$, thus the overhead depends on the longest link in the network.
For a LEO satellite at an altitude of 800~km and with an elevation angle of 70°, the distance at the cell edge, from~\cref{eq:distance}, is approximately 845~km, resulting in a propagation delay of 2.82 ms.
As a result, the guard period is up to 5.64 ms, which is more than the duration of 5 subframes in 5G \gls{nr}.

The excessive length of guard periods, where no transmissions should be scheduled to avoid interference, would lead to a very inefficient use of radio resources in satellite networks, which motivates our research toward more advanced slot allocation mechanisms, as described in~\cref{sec:ssa}.

\section{Enhanced Synchronized \\ Slot Allocation and Scheduling}
\label{sec:ssa}
\subsection{\acrfull{essa}}
\label{sub:essa}
One possible way to reduce the impact of the overhead in TDD is to allocate multiple slots of the same type (i.e., UL or DL), to reduce the number of required guard periods where no transmissions should be scheduled.
%This solution, however, constraints the flexibility of the system.

Along these lines, in this paper we propose \acrfull{essa}, which is designed to allocate multiple \gls{dl} slots during guard periods, before the corresponding \gls{ul} slot, if they do not interfere with other concurrent transmissions.
Let $\tau_m$ be the propagation delay of the shortest link in the network, while $t_{\rm UL}$ is the duration of \gls{ul} transmissions. Then, we define $T_{th}$~as
\begin{equation}
T_{th} = 2(\tau_M - \tau_m) + t_{\rm UL}.
\end{equation}

\begin{lem}
In TDD, if a \gls{dl} transmission is completed at time $t_0$, additional concurrent \gls{dl} transmissions scheduled by the satellite between $t_0 + T_{th}$ and $t_0 + 2\tau_M$ do not create interference, and can be successfully received by all \glspl{ue} within the coverage area of the satellite.
\end{lem}

\begin{proof}
Assume that an additional \gls{dl} slot (DL 2 in~\cref{fig:proposed}), scheduled after $T_{th}$ from the first \gls{dl} slot (DL 1 in~\cref{fig:proposed}), is received by \gls{ue} $i$ after $T_{th} + \tau_i$.
Moreover, assume that the \gls{ul} slot relative to the first \gls{dl} slot (UL 1 in~\cref{fig:proposed}) is transmitted by \gls{ue} $i$ at $t_0+ \tau_i+2(\tau_M - \tau_i)$, for a duration of $t_{\rm UL}$.
Therefore, the additional \gls{dl} slot does not interfere if
\begin{equation}
2(\tau_M - \tau_m) + t_{\rm UL} + \tau_i \geq  \tau_i+2(\tau_M - \tau_i) + t_{\rm UL},
\end{equation}
which is always verified if $\tau_m \leq \tau_i \; \forall i$, which is true since $\tau_m = \min_i (\tau_i)$.
\end{proof}

Indeed, additional \gls{dl} slots can be scheduled after a time interval equal to $2(\tau_M - \tau_m) + t_{\rm UL}$ from the end of the previous \gls{dl} slot, provided that $2\tau_m \geq t_{\rm UL}$. 
This inequality is generally verified in \gls{ntn} scenarios.
For example, for a LEO satellite at an altitude of 600 km, the minimum one-way propagation delay (at the zenith) is approximately $\tau_m=2$~ms.
As a result, the condition is satisfied if the duration of \gls{ul} transmissions is $t_{\rm UL}\leq 4$ ms, corresponding to 4 subframes in 5G \gls{nr}, which is a reasonable assumption in most network configurations.
This approach permits to schedule resources that would otherwise be unused because of guard periods, thereby reducing the communication overhead.

As shown in~\cref{fig:proposed}, the second \gls{dl} slot ({DL 2}) schedules 
another \gls{ul} slot after $2\tau_M$. If there are available slots after a time interval equal to $T_{th}$ from the end of {DL 2}, other \gls{dl} slots can be scheduled.
For instance, {DL 3} is scheduled immediately after the reception of {UL 1}.
This scheduling process continues iteratively until all packets are scheduled.
%Besides, multiple \gls{dl} slots can be allocated between $t_0 + T_{th}$ and the corresponding \gls{ul} slot in order to further increase the number of allocated slots, and therefore the DL throughput.
This approach, however, would inevitably reduce the \gls{ul} capacity in favor of the DL capacity since a lower number of \gls{ul} slots can be allocated, as we will evaluate in~\cref{sec:simulation_results}.

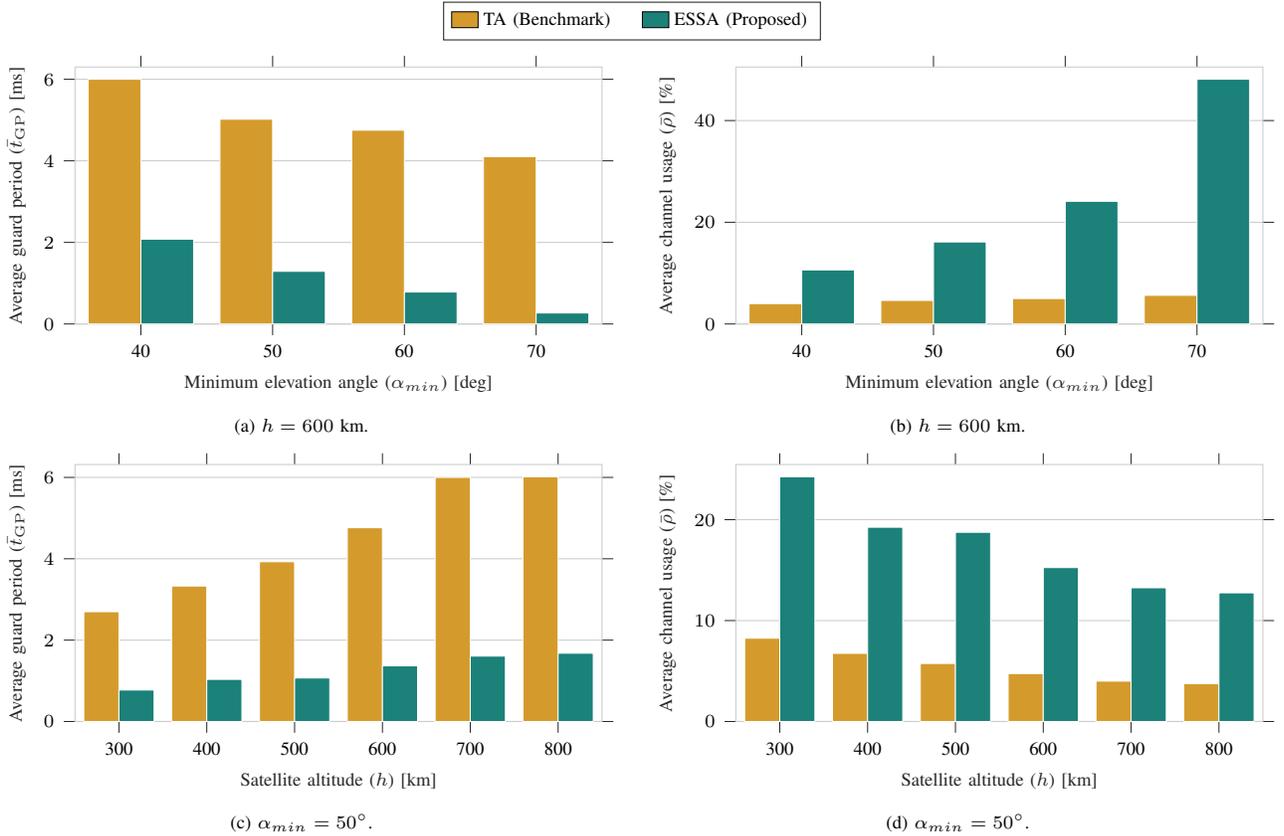
\begin{figure*}[t!]
    \begin{subfigure}[b]{\linewidth}
	\centering

\definecolor{black25}{RGB}{25,25,25}
\definecolor{mediumaquamarine102194165}{RGB}{102,194,165}
\definecolor{limegreen3122331}{RGB}{31,223,31}
\definecolor{silver}{RGB}{192,192,192}
\definecolor{lightgreen178223138}{RGB}{178,223,138}
\definecolor{lightsteelblue173203219}{RGB}{173,203,219}
\definecolor{steelblue31120180}{RGB}{31,120,180}

%\definecolor{color2}{RGB}{255,166,0}
%\definecolor{color1}{RGB}{0,100,160}

\definecolor{color2}{RGB}{27, 129, 121}
\definecolor{color1}{RGB}{212, 155, 44}

\begin{tikzpicture}
\pgfplotsset{every tick label/.append style={font=\scriptsize}}

\pgfplotsset{compat=1.11,
	/pgfplots/ybar legend/.style={
		/pgfplots/legend image code/.code={%
			\draw[##1,/tikz/.cd,yshift=-0.25em]
			(0cm,0cm) rectangle (10pt,0.6em);},
	},
}

\begin{axis}[%
width=0,
height=0,
at={(0,0)},
scale only axis,
xmin=0,
xmax=0,
xtick={},
ymin=0,
ymax=0,
ytick={},
axis background/.style={fill=white},
legend style={legend cell align=left,
              align=center,
              draw=white!15!black,
              at={(0.5, 1.3)},
              anchor=center,
              /tikz/every even column/.append style={column sep=1em}},
legend columns=2,
]
\addplot[ybar,ybar legend,draw=black,fill=color1,line width=0.08pt]
table[row sep=crcr]{%
	0	0\\
};
\addlegendentry{{TA (Benchmark)}}

\addplot[ybar legend,ybar,draw=black,fill=color2,line width=0.08pt]
  table[row sep=crcr]{%
	0	0\\
};
\addlegendentry{{ESSA (Proposed)}}

\end{axis}
\end{tikzpicture}%
	\end{subfigure}
 \vskip 0.2cm
    \centering
    \subfloat[][$h = 600$~km.]
	{
	    \label{fig:switch_duration_alpha}
        % This file was created with tikzplotlib v0.10.1.
\begin{tikzpicture}

\definecolor{darkslategray38}{RGB}{38,38,38}
\definecolor{darkslategray66}{RGB}{66,66,66}
\definecolor{lightgray204}{RGB}{204,204,204}
%\definecolor{color2}{RGB}{255,166,0}
%\definecolor{color1}{RGB}{0,100,160}
\definecolor{color2}{RGB}{27, 129, 121}
\definecolor{color1}{RGB}{212, 155, 44}

\begin{axis}[
width = \textwidth/2.1,
height = 5cm,
axis line style={lightgray204},
tick align=outside,
unbounded coords=jump,
x grid style={lightgray204},
xlabel=\textcolor{darkslategray38}{Minimum elevation angle ($\alpha_{min}$) [deg]},
xmajorticks=true,
xmin=-0.5, xmax=3.5,
xtick style={color=darkslategray38},
xtick={0,1,2,3},
xticklabels={40,50,60,70},
y grid style={lightgray204},
ylabel=\textcolor{darkslategray38}{Average guard period ($\Bar{t}_{\rm GP}$) [ms]},
ymajorgrids,
ymajorticks=true,
ymin=0, ymax=6.3,
ytick style={color=darkslategray38}
]
\draw[draw=white,fill=color1,line width=0.32pt] (axis cs:-0.4,0) rectangle (axis cs:0,6);
\draw[draw=white,fill=color1,line width=0.32pt] (axis cs:0.6,0) rectangle (axis cs:1,5.01973684210526);
\draw[draw=white,fill=color1,line width=0.32pt] (axis cs:1.6,0) rectangle (axis cs:2,4.75);
\draw[draw=white,fill=color1,line width=0.32pt] (axis cs:2.6,0) rectangle (axis cs:3,4.10326086956522);
\draw[draw=white,fill=color2,line width=0.32pt] (axis cs:-2.77555756156289e-17,0) rectangle (axis cs:0.4,2.07848837209302);

\draw[draw=white,fill=color2,line width=0.32pt] (axis cs:1,0) rectangle (axis cs:1.4,1.29038461538462);
\draw[draw=white,fill=color2,line width=0.32pt] (axis cs:2,0) rectangle (axis cs:2.4,0.782216494845361);
\draw[draw=white,fill=color2,line width=0.32pt] (axis cs:3,0) rectangle (axis cs:3.4,0.268782383419689);
\addplot [line width=0.864pt, darkslategray66, forget plot]
table {%
-0.2 nan
-0.2 nan
};
\addplot [line width=0.864pt, darkslategray66, forget plot]
table {%
0.8 nan
0.8 nan
};
\addplot [line width=0.864pt, darkslategray66, forget plot]
table {%
1.8 nan
1.8 nan
};
\addplot [line width=0.864pt, darkslategray66, forget plot]
table {%
2.8 nan
2.8 nan
};
\addplot [line width=0.864pt, darkslategray66, forget plot]
table {%
0.2 nan
0.2 nan
};
\addplot [line width=0.864pt, darkslategray66, forget plot]
table {%
1.2 nan
1.2 nan
};
\addplot [line width=0.864pt, darkslategray66, forget plot]
table {%
2.2 nan
2.2 nan
};
\addplot [line width=0.864pt, darkslategray66, forget plot]
table {%
3.2 nan
3.2 nan
};
\end{axis}

\end{tikzpicture}
	}
   \subfloat[][$h = 600$~km.]
	{
		\label{fig:channel_usage_alpha}
		% This file was created with tikzplotlib v0.10.1.
\begin{tikzpicture}

\definecolor{darkslategray38}{RGB}{38,38,38}
\definecolor{darkslategray66}{RGB}{66,66,66}
\definecolor{lightgray204}{RGB}{204,204,204}
\definecolor{color2}{RGB}{255,166,0}
\definecolor{color1}{RGB}{0,100,160}
\definecolor{color2}{RGB}{27, 129, 121}
\definecolor{color1}{RGB}{212, 155, 44}

\begin{axis}[
width = \textwidth/2.1,
height = 5cm,
axis line style={lightgray204},
tick align=outside,
unbounded coords=jump,
x grid style={lightgray204},
xlabel=\textcolor{darkslategray38}{Minimum elevation angle ($\alpha_{min}$) [deg]},
xmajorticks=true,
xmin=-0.5, xmax=3.5,
xtick style={color=darkslategray38},
xtick={0,1,2,3},
xticklabels={40,50,60,70},
y grid style={lightgray204},
ylabel=\textcolor{darkslategray38}{Average channel usage ($\Bar{\rho}$) [\%]},
ymajorgrids,
ymajorticks=true,
ymin=0, ymax=50.53125,
ytick style={color=darkslategray38}
]
\draw[draw=white,fill=color1,line width=0.32pt] (axis cs:-0.4,0) rectangle (axis cs:0,4);

\draw[draw=white,fill=color1,line width=0.32pt] (axis cs:0.6,0) rectangle (axis cs:1,4.625);
\draw[draw=white,fill=color1,line width=0.32pt] (axis cs:1.6,0) rectangle (axis cs:2,5);
\draw[draw=white,fill=color1,line width=0.32pt] (axis cs:2.6,0) rectangle (axis cs:3,5.625);
\draw[draw=white,fill=color2,line width=0.32pt] (axis cs:-2.77555756156289e-17,0) rectangle (axis cs:0.4,10.625);

\draw[draw=white,fill=color2,line width=0.32pt] (axis cs:1,0) rectangle (axis cs:1.4,16.125);
\draw[draw=white,fill=color2,line width=0.32pt] (axis cs:2,0) rectangle (axis cs:2.4,24.125);
\draw[draw=white,fill=color2,line width=0.32pt] (axis cs:3,0) rectangle (axis cs:3.4,48.125);
\addplot [line width=0.864pt, darkslategray66, forget plot]
table {%
-0.2 nan
-0.2 nan
};
\addplot [line width=0.864pt, darkslategray66, forget plot]
table {%
0.8 nan
0.8 nan
};
\addplot [line width=0.864pt, darkslategray66, forget plot]
table {%
1.8 nan
1.8 nan
};
\addplot [line width=0.864pt, darkslategray66, forget plot]
table {%
2.8 nan
2.8 nan
};
\addplot [line width=0.864pt, darkslategray66, forget plot]
table {%
0.2 nan
0.2 nan
};
\addplot [line width=0.864pt, darkslategray66, forget plot]
table {%
1.2 nan
1.2 nan
};
\addplot [line width=0.864pt, darkslategray66, forget plot]
table {%
2.2 nan
2.2 nan
};
\addplot [line width=0.864pt, darkslategray66, forget plot]
table {%
3.2 nan
3.2 nan
};
\end{axis}

\end{tikzpicture}
	}
    \vskip 0.2cm
    \subfloat[][ $\alpha_{min}=50^\circ$.]
	{
	    \label{fig:switch_duration_altitude}
        % This file was created with tikzplotlib v0.10.1.
\begin{tikzpicture}

\definecolor{darkslategray38}{RGB}{38,38,38}
\definecolor{darkslategray66}{RGB}{66,66,66}
\definecolor{lightgray204}{RGB}{204,204,204}
\definecolor{color2}{RGB}{255,166,0}
\definecolor{color1}{RGB}{0,100,160}
\definecolor{color2}{RGB}{27, 129, 121}
\definecolor{color1}{RGB}{212, 155, 44}

\begin{axis}[
width = \textwidth/2.1,
height = 5cm,
axis line style={lightgray204},
tick align=outside,
unbounded coords=jump,
x grid style={lightgray204},
xlabel=\textcolor{darkslategray38}{Satellite altitude ($h$) [km]},
xmajorticks=true,
xmin=-0.5, xmax=5.5,
xtick style={color=darkslategray38},
xtick={0,1,2,3,4,5},
xticklabels={300,400,500,600,700,800},
y grid style={lightgray204},
ylabel=\textcolor{darkslategray38}{Average guard period ($\Bar{t}_{\rm GP}$) [ms]},
ymajorgrids,
ymajorticks=true,
ymin=0, ymax=6.31640625,
ytick style={color=darkslategray38}
]
\draw[draw=white,fill=color1,line width=0.32pt] (axis cs:-0.4,0) rectangle (axis cs:0,2.69852941176471);

\draw[draw=white,fill=color1,line width=0.32pt] (axis cs:0.6,0) rectangle (axis cs:1,3.33035714285714);
\draw[draw=white,fill=color1,line width=0.32pt] (axis cs:1.6,0) rectangle (axis cs:2,3.92708333333333);
\draw[draw=white,fill=color1,line width=0.32pt] (axis cs:2.6,0) rectangle (axis cs:3,4.7625);
\draw[draw=white,fill=color1,line width=0.32pt] (axis cs:3.6,0) rectangle (axis cs:4,6);
\draw[draw=white,fill=color1,line width=0.32pt] (axis cs:4.6,0) rectangle (axis cs:5,6.015625);
\draw[draw=white,fill=color2,line width=0.32pt] (axis cs:-2.77555756156289e-17,0) rectangle (axis cs:0.4,0.772959183673469);

\draw[draw=white,fill=color2,line width=0.32pt] (axis cs:1,0) rectangle (axis cs:1.4,1.03525641025641);
\draw[draw=white,fill=color2,line width=0.32pt] (axis cs:2,0) rectangle (axis cs:2.4,1.06907894736842);
\draw[draw=white,fill=color2,line width=0.32pt] (axis cs:3,0) rectangle (axis cs:3.4,1.36693548387097);
\draw[draw=white,fill=color2,line width=0.32pt] (axis cs:4,0) rectangle (axis cs:4.4,1.60648148148148);
\draw[draw=white,fill=color2,line width=0.32pt] (axis cs:5,0) rectangle (axis cs:5.4,1.67788461538462);
\addplot [line width=0.864pt, darkslategray66, forget plot]
table {%
-0.2 nan
-0.2 nan
};
\addplot [line width=0.864pt, darkslategray66, forget plot]
table {%
0.8 nan
0.8 nan
};
\addplot [line width=0.864pt, darkslategray66, forget plot]
table {%
1.8 nan
1.8 nan
};
\addplot [line width=0.864pt, darkslategray66, forget plot]
table {%
2.8 nan
2.8 nan
};
\addplot [line width=0.864pt, darkslategray66, forget plot]
table {%
3.8 nan
3.8 nan
};
\addplot [line width=0.864pt, darkslategray66, forget plot]
table {%
4.8 nan
4.8 nan
};
\addplot [line width=0.864pt, darkslategray66, forget plot]
table {%
0.2 nan
0.2 nan
};
\addplot [line width=0.864pt, darkslategray66, forget plot]
table {%
1.2 nan
1.2 nan
};
\addplot [line width=0.864pt, darkslategray66, forget plot]
table {%
2.2 nan
2.2 nan
};
\addplot [line width=0.864pt, darkslategray66, forget plot]
table {%
3.2 nan
3.2 nan
};
\addplot [line width=0.864pt, darkslategray66, forget plot]
table {%
4.2 nan
4.2 nan
};
\addplot [line width=0.864pt, darkslategray66, forget plot]
table {%
5.2 nan
5.2 nan
};
\end{axis}

\end{tikzpicture}
	}
   \subfloat[][ $\alpha_{min}=50^\circ$.]
	{
		\label{fig:channel_usage_altitude}
		% This file was created with tikzplotlib v0.10.1.
\begin{tikzpicture}

\definecolor{darkslategray38}{RGB}{38,38,38}
\definecolor{darkslategray66}{RGB}{66,66,66}
\definecolor{lightgray204}{RGB}{204,204,204}
\definecolor{color2}{RGB}{255,166,0}
\definecolor{color1}{RGB}{0,100,160}
\definecolor{color2}{RGB}{27, 129, 121}
\definecolor{color1}{RGB}{212, 155, 44}

\begin{axis}[
width = \textwidth/2.1,
height = 5cm,
axis line style={lightgray204},
tick align=outside,
unbounded coords=jump,
x grid style={lightgray204},
xlabel=\textcolor{darkslategray38}{Satellite altitude ($h$) [km]},
xmajorticks=true,
xmin=-0.5, xmax=5.5,
xtick style={color=darkslategray38},
xtick={0,1,2,3,4,5},
xticklabels={300,400,500,600,700,800},
y grid style={lightgray204},
ylabel=\textcolor{darkslategray38}{Average channel usage ($\Bar{\rho}$) [\%]},
ymajorgrids,
ymajorticks=true,
ymin=0, ymax=25.4625,
ytick style={color=darkslategray38}
]
\draw[draw=white,fill=color1,line width=0.32pt] (axis cs:-0.4,0) rectangle (axis cs:0,8.25);

\draw[draw=white,fill=color1,line width=0.32pt] (axis cs:0.6,0) rectangle (axis cs:1,6.75);
\draw[draw=white,fill=color1,line width=0.32pt] (axis cs:1.6,0) rectangle (axis cs:2,5.75);
\draw[draw=white,fill=color1,line width=0.32pt] (axis cs:2.6,0) rectangle (axis cs:3,4.75);
\draw[draw=white,fill=color1,line width=0.32pt] (axis cs:3.6,0) rectangle (axis cs:4,4);
\draw[draw=white,fill=color1,line width=0.32pt] (axis cs:4.6,0) rectangle (axis cs:5,3.75);
\draw[draw=white,fill=color2,line width=0.32pt] (axis cs:-2.77555756156289e-17,0) rectangle (axis cs:0.4,24.25);

\draw[draw=white,fill=color2,line width=0.32pt] (axis cs:1,0) rectangle (axis cs:1.4,19.25);
\draw[draw=white,fill=color2,line width=0.32pt] (axis cs:2,0) rectangle (axis cs:2.4,18.75);
\draw[draw=white,fill=color2,line width=0.32pt] (axis cs:3,0) rectangle (axis cs:3.4,15.25);
\draw[draw=white,fill=color2,line width=0.32pt] (axis cs:4,0) rectangle (axis cs:4.4,13.25);
\draw[draw=white,fill=color2,line width=0.32pt] (axis cs:5,0) rectangle (axis cs:5.4,12.75);
\addplot [line width=0.864pt, darkslategray66, forget plot]
table {%
-0.2 nan
-0.2 nan
};
\addplot [line width=0.864pt, darkslategray66, forget plot]
table {%
0.8 nan
0.8 nan
};
\addplot [line width=0.864pt, darkslategray66, forget plot]
table {%
1.8 nan
1.8 nan
};
\addplot [line width=0.864pt, darkslategray66, forget plot]
table {%
2.8 nan
2.8 nan
};
\addplot [line width=0.864pt, darkslategray66, forget plot]
table {%
3.8 nan
3.8 nan
};
\addplot [line width=0.864pt, darkslategray66, forget plot]
table {%
4.8 nan
4.8 nan
};
\addplot [line width=0.864pt, darkslategray66, forget plot]
table {%
0.2 nan
0.2 nan
};
\addplot [line width=0.864pt, darkslategray66, forget plot]
table {%
1.2 nan
1.2 nan
};
\addplot [line width=0.864pt, darkslategray66, forget plot]
table {%
2.2 nan
2.2 nan
};
\addplot [line width=0.864pt, darkslategray66, forget plot]
table {%
3.2 nan
3.2 nan
};
\addplot [line width=0.864pt, darkslategray66, forget plot]
table {%
4.2 nan
4.2 nan
};
\addplot [line width=0.864pt, darkslategray66, forget plot]
table {%
5.2 nan
5.2 nan
};
\end{axis}

\end{tikzpicture}
	}
 \caption{Average guard period ($\Bar{t}_{\rm GP}$) and channel usage ($\Bar{\rho}$) vs. $\alpha_{min}$ and $h$.}
 \label{fig:alpha}
 \end{figure*}

\subsection{Scheduling}
\label{sub:scheduling}

In order to validate the performance of ESSA, we introduce two scheduling methods that select a subset of $N_{s}\leq N_{\rm UE}$ UEs that should transmit to maximize the Ergodic capacity of the cell. We call $S$ the set of selected UEs and $U$ the set of available UEs.
%We assume that each UE is allocated to an equal number of symbols from the satellite.
We consider the following scheduling methods.
\begin{itemize}
   \item MG (Benchmark):
   %In this scheme, 
   %slot allocation is based on the benchmark TDD frame structure with TA described in~\cref{sub:frame}, i.e., with no optimization of guard periods. Then, 
   The satellite selects the $N_s$ \glspl{ue} with the highest \gls{snr}, i.e.,
    \begin{subequations}
        \begin{align}
        & S = \underset{S}{\argmax}
        & & \left( \gamma_S \right) \\
        & \text{subject to}
        & & \gamma_S = \min_{i\in S} (\gamma_i), \\
        &&& |S| = N_{s}, \\
        &&& S \subseteq U.
        \end{align}
    \label{eq:S_MG}
    \end{subequations}
   Therefore, MG selects UEs with the best channel conditions in an attempt to maximize the network~capacity. 
   %However, MG-TA is affected by the limitations of TDD in terms of overhead and resource utilization.
       
   \item MS:
   %slot allocation is based on ESSA, and 
   The satellite selects the $N_s$ \glspl{ue} with the minimum differential propagation delay, i.e.,
    \begin{subequations}
        \begin{align}
        & S = \underset{S}{\argmin}
        & & \left( \tau_S \right) \\
        & \text{subject to}
        & & \tau_S = \max_{i,j\in S} | \tau_i - \tau_j |, \\
        &&& |S| = N_{s}, \\
        &&& S \subseteq U.
        \end{align}
    \label{eq:S_MG}
    \label{eq:S_MS}
    \end{subequations}
   Therefore,
   MS prioritizes \glspl{ue} with similar propagation delays, regardless of the actual channel conditions. 
   %in the attempt to reduce $\tau_m$ and increase the utilization of radio resources, which permits to also improve the channel capacity. 
 \end{itemize}
These scheduling methods can be combined with ESSA for allocating resources in TDD, as we will demonstrate and evaluate in~\cref{sub:results-scheduling}.

\section{Simulation results}
\label{sec:simulation_results}
The simulation parameters are reported in~\cref{tab:parameters}. Specifically, our scenario consists of a LEO satellite at a fixed altitude $h\in\{300,\dots,800\}$ km, while the elevation angle $\alpha_i$ of a generic UE $i$ is modeled according to a uniform distribution in $[\alpha_{min},\,\alpha_{max}]$, where $\alpha_{min}$ and $\alpha_{max}$ represent the minimum and maximum elevation angles, respectively.
We consider UL and DL transmissions in the Ka-bands, i.e., in the \gls{mmwave} spectrum~\cite{giordani2020satellite}, at frequency $f=28$ GHz and with a bandwidth $B=200$~MHz.

\begin{table}[t!]
\renewcommand{\arraystretch}{1.3}
\footnotesize
\centering
\caption{Simulation parameters.}
\label{tab:parameters}
\begin{tabular}{|l|l|}
\hline
\textbf{Parameter}                              & \textbf{Values}                 \\ \hline
Transmit power ($P^{tx}$) [dBW]                        & $-6$                           \\
Total antenna gain ($G^{tx}+G^{rx}$) [dBi]                       & 24                 \\
Bandwidth ($B$) [MHz]                         & 200                            \\
Carrier frequency ($f$) [GHz]                        & 28                           \\
Noise temperature ($T$) [K]                        & 290                           \\
Noise figure ($N_f$) [dB]                        & 5                           \\
Number of UEs ($N_{\rm UE}$)                        & 100                          \\
Selectable UEs ($N_{s}$)                        & 10                          \\
Maximum elevation angle ($\alpha_{max}$) [deg]     & 90                          \\
Radius of the Earth ($R_E$) [km]     & 6371                          \\
Channel parameters ($\{A_g,A_s,\sigma_s\}$) & \cite{38811}   \\
\hline
\end{tabular}
\vspace{-1.2em}
\end{table}

Results are given in terms of the following metrics:
\begin{itemize}
    \item Average guard period ($\Bar{t}_{\rm GP}$), i.e., the average time between a \gls{dl} slot and the corresponding \gls{ul} slot, observed from the satellite reference~point.
    \item Average channel usage ($\Bar{\rho}$), i.e., the average percentage of allocated slots (for either \gls{ul} or \gls{dl} transmissions) over the total number of available slots. Therefore, it is an indication of the system's overhead. For instance, if $\rho = 0.5$, half of the slots are allocated for transmissions, while the rest is used as guard periods.
    \item Average capacity ($\Bar{C}$), i.e., the average of the Ergodic capacity of the $N_s$ selected UEs in the cell via MS or MG scheduling, where the Ergodic capacity of the single UE was defined in~\cref{eq:C}.
\end{itemize}

In~\cref{sub:results-slot} we will compare the performance of the baseline TDD slot allocation with TA described in~\cref{sub:frame}, vs. the proposed \gls{essa} scheme described in~\cref{sub:essa} where slot allocation is optimized for NTN. 
Then, in~\cref{sub:results-scheduling} we will compare MG and MS in terms of scheduling.

\subsection{Slot Allocation Results}
\label{sub:results-slot}
In this section we compare the performance of TA vs. ESSA in terms of the duration of guard periods and the channel usage. In TA, the guard period is constant and equal to $2\tau_{M}$, which is required to maintain synchronization via timing advance. In ESSA, additional DL slots can be scheduled during guard periods based on Lemma 1.

\subsubsection*{Minimum elevation angle}
First, in~\cref{fig:switch_duration_alpha,fig:channel_usage_alpha} we analyze the impact of the minimum elevation angle $\alpha_{min}$ with a LEO satellite at an altitude $h=600$~km.\footnote{We investigate the value of $\alpha_{min}$ in the range from $40^\circ$ to $70^\circ$, which is in line with most commercial satellite constellations, as described in~\cite{9473799}.}
In general, as $\alpha_{min}$ increases, $\Bar{t}_{\rm GP}$ decreases.
This occurs for TA because the length of the longest link in the network decreases as $\alpha_{min}$ increases, and so does the corresponding propagation delay $\tau_{M}$.
A similar trend is observed for {ESSA}, as $T_{th}$ also decreases as $\alpha_{min}$ increases.
In any case, ESSA outperforms TA.
For TA, even with $\alpha_{min}=70^\circ$, which guarantees a good alignment between the LEO satellite and the UEs on the ground, $\Bar{t}_{\rm GP}\simeq4$~ms, with $\Bar{\rho}\simeq 5.6\%$.

% In this scenario, with a slot duration of 0.25 ms assuming 5G NR numerology~3~\cite{38300}, only $0.25/8=3\%$ of the resources can be scheduled for transmission, while the rest will be used for the guard periods.

In contrast, for {ESSA} we have $\Bar{t}_{\rm GP}<1~$ms, and $\Bar{\rho}\simeq 48\%$. This is because ESSA can schedule multiple DL slots during guard periods before a UL slot, which permits a more efficient use of resources.

\subsubsection*{Satellite altitude}
In~\cref{fig:switch_duration_altitude,fig:channel_usage_altitude} we analyze the impact of the satellite altitude $h$, with $\alpha_{min}=50^\circ$.
As $h$ increases, both $\tau_{M}$ and $T_{th}$ increase, resulting in a longer guard period for both schemes.
At $h=300$~km, {ESSA} achieves a channel usage of nearly $25\%$, with an average guard period of less than 1~ms.
Even at $h=800$~km, {ESSA} maintains a channel usage of over $10\%$, i.e., almost three times higher than TA, with a guard period under 2~ms.

The above results demonstrate that ESSA can effectively reduce the overhead for slot allocation compared to TA even in those scenarios characterized by a very long propagation delay, and stands out as a promising approach for NTN.

\subsection{Scheduling Results}
\label{sub:results-scheduling}

We now integrate the two slot allocation schemes presented in~\cref{sub:frame} and~\cref{sub:essa} (TA and ESSA)  with the two scheduling methods presented in~\cref{sub:scheduling} (MG and MS). This results in the following combinations:
\begin{itemize}
    \item MG-TA and MS-TA: Slot allocation is based on the benchmark TDD frame structure with TA, i.e., with no optimization of guard periods. 
    \item MG-ESSA and MS-ESSA: Slot allocation is based on ESSA,  so additional DL slots can be allocated during guard periods if $2\tau_m \geq t_{\rm UL}$.
\end{itemize}

\begin{figure*}[t!]
	\begin{subfigure}[b]{\linewidth}
		\centering

\definecolor{black25}{RGB}{25,25,25}
\definecolor{mediumaquamarine102194165}{RGB}{102,194,165}
\definecolor{limegreen3122331}{RGB}{31,223,31}
\definecolor{silver}{RGB}{192,192,192}
\definecolor{lightgreen178223138}{RGB}{178,223,138}
\definecolor{lightsteelblue173203219}{RGB}{173,203,219}
\definecolor{steelblue31120180}{RGB}{31,120,180}

\definecolor{color2}{RGB}{27, 129, 121}
\definecolor{color1}{RGB}{212, 155, 44}

\begin{tikzpicture}
\pgfplotsset{every tick label/.append style={font=\scriptsize}}

\pgfplotsset{compat=1.11,
	/pgfplots/ybar legend/.style={
		/pgfplots/legend image code/.code={%
			\draw[##1,/tikz/.cd,yshift=-0.25em]
			(0cm,0cm) rectangle (10pt,0.6em);},
	},
}

\begin{axis}[%
width=0,
height=0,
at={(0,0)},
scale only axis,
xmin=0,
xmax=0,
xtick={},
ymin=0,
ymax=0,
ytick={},
axis background/.style={fill=white},
legend style={legend cell align=left,
              align=center,
              draw=white!15!black,
              at={(0.5, 1.3)},
              anchor=center,
              /tikz/every even column/.append style={column sep=1em}},
legend columns=4,
]
\addplot[ybar,ybar legend,draw=black,fill=color1,line width=0.08pt]
table[row sep=crcr]{%
	0	0\\
};
\addlegendentry{{TA (Benchmark)}}

\addplot[ybar legend,ybar,draw=black,fill=color2,line width=0.08pt]
  table[row sep=crcr]{%
	0	0\\
};
\addlegendentry{{ESSA}}

\addplot[ybar,ybar legend,draw=black,fill=white,line width=0.08pt,postaction={
	pattern=north east lines
}]
table[row sep=crcr]{%
	0	0\\
};
\addlegendentry{MG}

\addplot[ybar,ybar legend,draw=black,fill=white,line width=0.08pt]
table[row sep=crcr]{%
	0	0\\
};
\addlegendentry{MS}

\end{axis}
\end{tikzpicture}%
	\end{subfigure}
	\vskip 0.2cm
	\centering
    \subfloat[][DSU slot pattern.]
	{
	    \label{fig:snr}
        \input{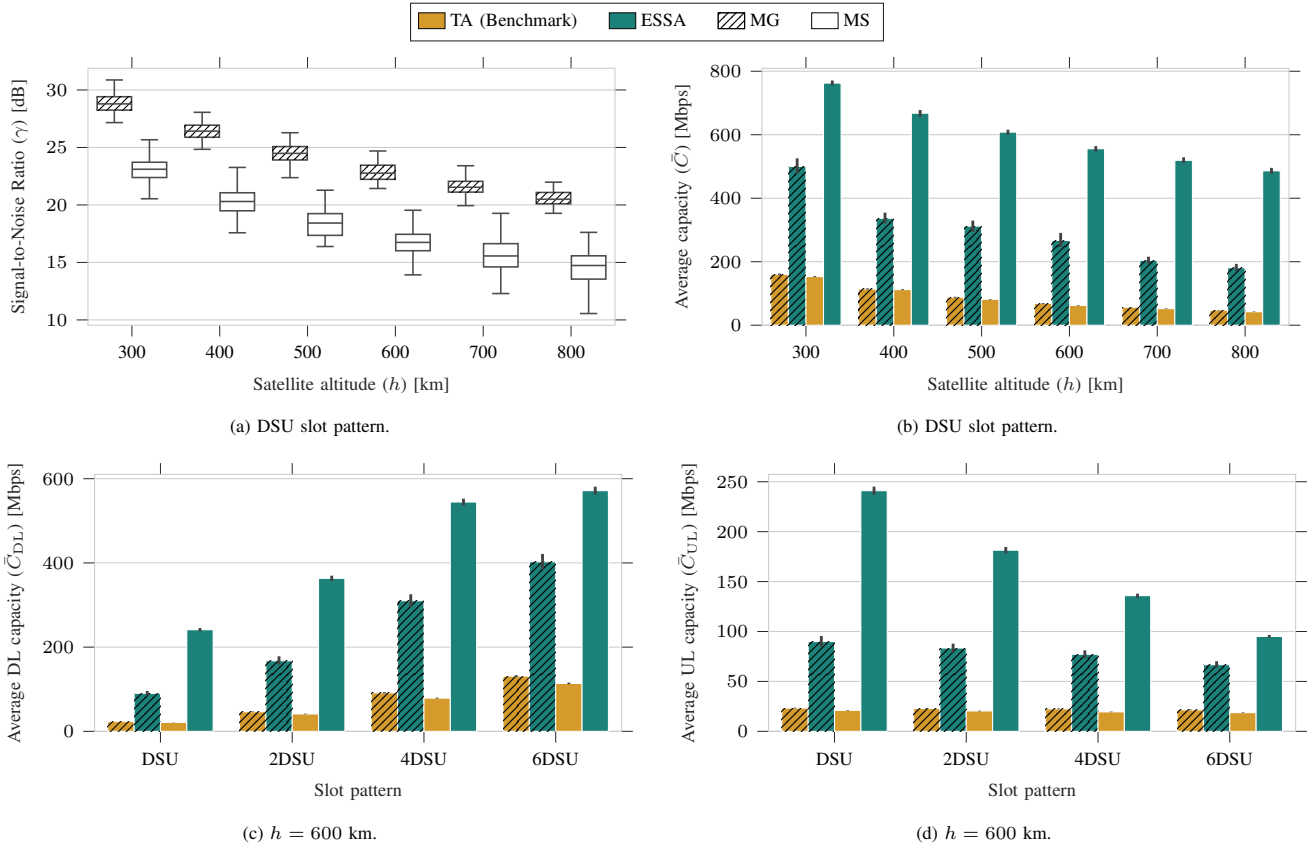}
	}
   \subfloat[][DSU slot pattern.]
	{
		\label{fig:capacity}
		% This file was created with tikzplotlib v0.10.1.
\begin{tikzpicture}

\definecolor{darkslategray38}{RGB}{38,38,38}
\definecolor{darkslategray66}{RGB}{66,66,66}
\definecolor{color2}{RGB}{181,92,95}
\definecolor{lightgray204}{RGB}{204,204,204}
\definecolor{color1}{RGB}{95,157,109}
\definecolor{color2}{RGB}{203,136,99}
\definecolor{color1}{RGB}{88,116,163}

\definecolor{color2}{RGB}{27, 129, 121}
\definecolor{color1}{RGB}{212, 155, 44}

\begin{axis}[
width = \textwidth/2.1,
height = 5cm,
axis line style={lightgray204},
legend cell align={left},
legend style={fill opacity=0.8, draw opacity=1, text opacity=1, draw=none},
tick align=outside,
x grid style={lightgray204},
xlabel=\textcolor{darkslategray38}{Satellite altitude ($h$) [km]},
xmajorticks=true,
xmin=-0.5, xmax=5.5,
xtick style={color=darkslategray38},
xtick={0,1,2,3,4,5},
xticklabels={300,400,500,600,700,800},
y grid style={lightgray204},
ylabel=\textcolor{darkslategray38}{Average capacity ($\Bar{C}$) [Mbps]},
ymajorgrids,
ymajorticks=true,
ymin=0, ymax=809.33717182893,
ytick style={color=darkslategray38}
]
\draw[draw=white,fill=color1,postaction={
	pattern=north east lines
}] (axis cs:-0.4,0) rectangle (axis cs:-0.2,160.019709633108);
%\addlegendentry{MGSSA}

\draw[draw=white,fill=color1,postaction={
	pattern=north east lines
}] (axis cs:0.6,0) rectangle (axis cs:0.8,114.120075783547);
\draw[draw=white,fill=color1,postaction={
	pattern=north east lines
}] (axis cs:1.6,0) rectangle (axis cs:1.8,88.0755828648173);
\draw[draw=white,fill=color1,postaction={
	pattern=north east lines
}] (axis cs:2.6,0) rectangle (axis cs:2.8,68.6663872738909);
\draw[draw=white,fill=color1,postaction={
	pattern=north east lines
}] (axis cs:3.6,0) rectangle (axis cs:3.8,56.5960757475899);
\draw[draw=white,fill=color1,postaction={
	pattern=north east lines
}] (axis cs:4.6,0) rectangle (axis cs:4.8,47.283828055388);
\draw[draw=white,fill=color2,postaction={
	pattern=north east lines
}] (axis cs:-0.2,0) rectangle (axis cs:0,500.049503142749);
%\addlegendimage{ybar,ybar legend,draw=white,fill=color2}
%\addlegendentry{MGESSA}

\draw[draw=white,fill=color2,postaction={
	pattern=north east lines
}] (axis cs:0.8,0) rectangle (axis cs:1,336.811888761798);
\draw[draw=white,fill=color2,postaction={
	pattern=north east lines
}] (axis cs:1.8,0) rectangle (axis cs:2,312.240834806827);
\draw[draw=white,fill=color2,postaction={
	pattern=north east lines
}] (axis cs:2.8,0) rectangle (axis cs:3,266.910599653066);
\draw[draw=white,fill=color2,postaction={
	pattern=north east lines
}] (axis cs:3.8,0) rectangle (axis cs:4,204.569126541663);
\draw[draw=white,fill=color2,postaction={
	pattern=north east lines
}] (axis cs:4.8,0) rectangle (axis cs:5,181.220349926584);
\draw[draw=white,fill=color1] (axis cs:2.77555756156289e-17,0) rectangle (axis cs:0.2,153.063848751154);
%\addlegendimage{ybar,ybar legend,draw=white,fill=color1}
%\addlegendentry{MSSSA}

\draw[draw=white,fill=color1] (axis cs:1,0) rectangle (axis cs:1.2,112.153049073238);
\draw[draw=white,fill=color1] (axis cs:2,0) rectangle (axis cs:2.2,81.6691624520778);
\draw[draw=white,fill=color1] (axis cs:3,0) rectangle (axis cs:3.2,62.4464811589893);
\draw[draw=white,fill=color1] (axis cs:4,0) rectangle (axis cs:4.2,52.3669398193515);
\draw[draw=white,fill=color1] (axis cs:5,0) rectangle (axis cs:5.2,42.5087625364674);
\draw[draw=white,fill=color2] (axis cs:0.2,0) rectangle (axis cs:0.4,762.640626402623);
%\addlegendimage{ybar,ybar legend,draw=white,fill=color2}
%\addlegendentry{MSESSA}

\draw[draw=white,fill=color2] (axis cs:1.2,0) rectangle (axis cs:1.4,667.706136006948);
\draw[draw=white,fill=color2] (axis cs:2.2,0) rectangle (axis cs:2.4,608.163459797076);
\draw[draw=white,fill=color2] (axis cs:3.2,0) rectangle (axis cs:3.4,555.857691034053);
\draw[draw=white,fill=color2] (axis cs:4.2,0) rectangle (axis cs:4.4,519.433748470819);
\draw[draw=white,fill=color2] (axis cs:5.2,0) rectangle (axis cs:5.4,486.270757570341);
\addplot [line width=1.08pt, darkslategray66, forget plot]
table {%
-0.3 157.661382880031
-0.3 162.589738189028
};
\addplot [line width=1.08pt, darkslategray66, forget plot]
table {%
0.7 112.465599747983
0.7 115.896197514279
};
\addplot [line width=1.08pt, darkslategray66, forget plot]
table {%
1.7 86.5072845938344
1.7 89.7402532320078
};
\addplot [line width=1.08pt, darkslategray66, forget plot]
table {%
2.7 67.3272110245876
2.7 69.9143543698518
};
\addplot [line width=1.08pt, darkslategray66, forget plot]
table {%
3.7 55.7059769241822
3.7 57.4756837600762
};
\addplot [line width=1.08pt, darkslategray66, forget plot]
table {%
4.7 46.5036765890836
4.7 48.0777525996271
};
\addplot [line width=1.08pt, darkslategray66, forget plot]
table {%
-0.1 477.732636127407
-0.1 525.640712320457
};
\addplot [line width=1.08pt, darkslategray66, forget plot]
table {%
0.9 320.769063163907
0.9 354.566831850465
};
\addplot [line width=1.08pt, darkslategray66, forget plot]
table {%
1.9 293.933642297828
1.9 329.68220525909
};
\addplot [line width=1.08pt, darkslategray66, forget plot]
table {%
2.9 247.065063743008
2.9 290.739587044148
};
\addplot [line width=1.08pt, darkslategray66, forget plot]
table {%
3.9 193.602526560925
3.9 215.77592816098
};
\addplot [line width=1.08pt, darkslategray66, forget plot]
table {%
4.9 170.897283659832
4.9 193.201353882355
};
\addplot [line width=1.08pt, darkslategray66, forget plot]
table {%
0.1 151.520664654481
0.1 154.625260756579
};
\addplot [line width=1.08pt, darkslategray66, forget plot]
table {%
1.1 110.692677325864
1.1 113.521144370647
};
\addplot [line width=1.08pt, darkslategray66, forget plot]
table {%
2.1 80.6697155531765
2.1 82.6561878228298
};
\addplot [line width=1.08pt, darkslategray66, forget plot]
table {%
3.1 61.5188530547921
3.1 63.3562906724787
};
\addplot [line width=1.08pt, darkslategray66, forget plot]
table {%
4.1 51.521860050708
4.1 53.23287158541
};
\addplot [line width=1.08pt, darkslategray66, forget plot]
table {%
5.1 41.7916216466195
5.1 43.3016086448925
};
\addplot [line width=1.08pt, darkslategray66, forget plot]
table {%
0.3 755.591714035203
0.3 770.797306503742
};
\addplot [line width=1.08pt, darkslategray66, forget plot]
table {%
1.3 656.036850734714
1.3 677.850222022575
};
\addplot [line width=1.08pt, darkslategray66, forget plot]
table {%
2.3 600.715565440132
2.3 615.733797021795
};
\addplot [line width=1.08pt, darkslategray66, forget plot]
table {%
3.3 547.328623012527
3.3 564.290535435848
};
\addplot [line width=1.08pt, darkslategray66, forget plot]
table {%
4.3 510.654264644319
4.3 528.848951920978
};
\addplot [line width=1.08pt, darkslategray66, forget plot]
table {%
5.3 476.591390709458
5.3 495.631467284209
};
\end{axis}

\end{tikzpicture}
	}
    \vskip 0.2cm
	\subfloat[][$h=600$~km.]
	{
		\label{fig:capacity_slot_pattern}
		% This file was created with tikzplotlib v0.10.1.
\begin{tikzpicture}

\definecolor{darkslategray38}{RGB}{38,38,38}
\definecolor{darkslategray66}{RGB}{66,66,66}
\definecolor{color2}{RGB}{181,92,95}
\definecolor{lightgray204}{RGB}{204,204,204}
\definecolor{color1}{RGB}{95,157,109}
\definecolor{color2,postaction={
	pattern=north east lines
}}{RGB}{203,136,99}
\definecolor{color1,postaction={
	pattern=north east lines
}}{RGB}{88,116,163}

\definecolor{color2}{RGB}{27, 129, 121}
\definecolor{color1}{RGB}{212, 155, 44}

\begin{axis}[
width = \textwidth/2.1,
height = 5cm,
axis line style={lightgray204},
tick align=outside,
x grid style={lightgray204},
xlabel=\textcolor{darkslategray38}{Slot pattern},
xmajorticks=true,
xmin=-0.5, xmax=3.5,
xtick style={color=darkslategray38},
xtick={0,1,2,3},
xticklabels={DSU,2DSU,4DSU,6DSU},
y grid style={lightgray204},
ylabel=\textcolor{darkslategray38}{Average DL capacity ($\Bar{C}_{\rm DL}$) [Mbps]},
ymajorgrids,
ymajorticks=true,
ymin=0, ymax=610.220888667633,
ytick style={color=darkslategray38},
ytick={0,200,400,600,800},
yticklabels={
  \(\displaystyle {0}\),
  \(\displaystyle {200}\),
  \(\displaystyle {400}\),
  \(\displaystyle {600}\),
  \(\displaystyle {800}\)
}
]
\draw[draw=white,fill=color1,postaction={
	pattern=north east lines
}] (axis cs:-0.4,0) rectangle (axis cs:-0.2,23.6449917406759);
%\addlegendentry{MGSSA}

\draw[draw=white,fill=color1,postaction={
	pattern=north east lines
}] (axis cs:0.6,0) rectangle (axis cs:0.8,46.5171819755372);
\draw[draw=white,fill=color1,postaction={
	pattern=north east lines
}] (axis cs:1.6,0) rectangle (axis cs:1.8,91.3764096094603);
\draw[draw=white,fill=color1,postaction={
	pattern=north east lines
}] (axis cs:2.6,0) rectangle (axis cs:2.8,131.138441565818);
\draw[draw=white,fill=color2,postaction={
	pattern=north east lines
}] (axis cs:-0.2,0) rectangle (axis cs:0,90.2687021864286);
%\addlegendentry{MGESSA}

\draw[draw=white,fill=color2,postaction={
	pattern=north east lines
}] (axis cs:0.8,0) rectangle (axis cs:1,167.563362724335);
\draw[draw=white,fill=color2,postaction={
	pattern=north east lines
}] (axis cs:1.8,0) rectangle (axis cs:2,310.142469368833);
\draw[draw=white,fill=color2,postaction={
	pattern=north east lines
}] (axis cs:2.8,0) rectangle (axis cs:3,402.661340768263);
\draw[draw=white,fill=color1] (axis cs:2.77555756156289e-17,0) rectangle (axis cs:0.2,21.2559071156417);
%\addlegendimage{ybar,ybar legend,draw=white,fill=color1}
%\addlegendentry{MSSSA}

\draw[draw=white,fill=color1] (axis cs:1,0) rectangle (axis cs:1.2,41.3042827917211);
\draw[draw=white,fill=color1] (axis cs:2,0) rectangle (axis cs:2.2,78.9031457603919);
\draw[draw=white,fill=color1] (axis cs:3,0) rectangle (axis cs:3.2,113.80702546916);
\draw[draw=white,fill=color2] (axis cs:0.2,0) rectangle (axis cs:0.4,241.260073337324);
%\addlegendimage{ybar,ybar legend,draw=white,fill=color2}
%\addlegendentry{MSESSA}

\draw[draw=white,fill=color2] (axis cs:1.2,0) rectangle (axis cs:1.4,363.35884922468);
\draw[draw=white,fill=color2] (axis cs:2.2,0) rectangle (axis cs:2.4,544.611763607685);
\draw[draw=white,fill=color2] (axis cs:3.2,0) rectangle (axis cs:3.4,571.550386532349);
\addplot [line width=1.08pt, darkslategray66, forget plot]
table {%
-0.3 23.2661514361012
-0.3 24.0534574477208
};
\addplot [line width=1.08pt, darkslategray66, forget plot]
table {%
0.7 45.7868131797711
0.7 47.222296940286
};
\addplot [line width=1.08pt, darkslategray66, forget plot]
table {%
1.7 89.8667293411493
1.7 92.923879451543
};
\addplot [line width=1.08pt, darkslategray66, forget plot]
table {%
2.7 128.902147471924
2.7 133.399381903072
};
\addplot [line width=1.08pt, darkslategray66, forget plot]
table {%
-0.1 85.1391561657998
-0.1 95.6571184714603
};
\addplot [line width=1.08pt, darkslategray66, forget plot]
table {%
0.9 158.734633362045
0.9 178.242268239787
};
\addplot [line width=1.08pt, darkslategray66, forget plot]
table {%
1.9 295.085405830794
1.9 325.516697283604
};
\addplot [line width=1.08pt, darkslategray66, forget plot]
table {%
2.9 384.905462814903
2.9 421.382936516879
};
\addplot [line width=1.08pt, darkslategray66, forget plot]
table {%
0.1 20.9108604269059
0.1 21.5750284528079
};
\addplot [line width=1.08pt, darkslategray66, forget plot]
table {%
1.1 40.6275096001533
1.1 42.0245193598582
};
\addplot [line width=1.08pt, darkslategray66, forget plot]
table {%
2.1 77.8320145488739
2.1 80.0831823926124
};
\addplot [line width=1.08pt, darkslategray66, forget plot]
table {%
3.1 111.995947788654
3.1 115.704102769142
};
\addplot [line width=1.08pt, darkslategray66, forget plot]
table {%
0.3 237.227492807622
0.3 245.356411168637
};
\addplot [line width=1.08pt, darkslategray66, forget plot]
table {%
1.3 357.275102656015
1.3 369.714579736654
};
\addplot [line width=1.08pt, darkslategray66, forget plot]
table {%
2.3 536.19257599272
2.3 552.5938457467
};
\addplot [line width=1.08pt, darkslategray66, forget plot]
table {%
3.3 562.239397758445
3.3 581.162751112031
};
\end{axis}

\end{tikzpicture}
	}
	\subfloat[][$h=600$~km.]
	{
		\label{fig:tx_slot_pattern}
		% This file was created with tikzplotlib v0.10.1.
\begin{tikzpicture}

\definecolor{darkslategray38}{RGB}{38,38,38}
\definecolor{darkslategray66}{RGB}{66,66,66}
\definecolor{color2}{RGB}{181,92,95}
\definecolor{lightgray204}{RGB}{204,204,204}
\definecolor{color1}{RGB}{95,157,109}
\definecolor{}{RGB}{203,136,99}
\definecolor{color1,postaction={
	pattern=north east lines
}}{RGB}{88,116,163}

\definecolor{color2}{RGB}{27, 129, 121}
\definecolor{color1}{RGB}{212, 155, 44}

\begin{axis}[
width = \textwidth/2.1,
height = 5cm,
axis line style={lightgray204},
legend cell align={left},
legend style={fill opacity=0.8, draw opacity=1, text opacity=1, draw=none},
tick align=outside,
x grid style={lightgray204},
xlabel=\textcolor{darkslategray38}{Slot pattern},
xmajorticks=true,
xmin=-0.5, xmax=3.5,
xtick style={color=darkslategray38},
xtick={0,1,2,3},
xticklabels={DSU,2DSU,4DSU,6DSU},
y grid style={lightgray204},
ylabel=\textcolor{darkslategray38}{Average UL capacity ($\Bar{C}_{\rm UL}$) [Mbps]},
ymajorgrids,
ymajorticks=true,
ymin=0, ymax=257.348803896006,
ytick style={color=darkslategray38},
ytick={0,50,100,150,200,250,300},
yticklabels={
  \(\displaystyle {0}\),
  \(\displaystyle {50}\),
  \(\displaystyle {100}\),
  \(\displaystyle {150}\),
  \(\displaystyle {200}\),
  \(\displaystyle {250}\),
  \(\displaystyle {300}\)
}
]
\draw[draw=white,fill=color1,postaction={
	pattern=north east lines
}] (axis cs:-0.4,0) rectangle (axis cs:-0.2,23.3125214697867);

%\addlegendentry{MGSSA}

\draw[draw=white,fill=color1,postaction={
	pattern=north east lines
}] (axis cs:0.6,0) rectangle (axis cs:0.8,22.9184416602086);
\draw[draw=white,fill=color1,postaction={
	pattern=north east lines
}] (axis cs:1.6,0) rectangle (axis cs:1.8,22.5041606735326);
\draw[draw=white,fill=color1,postaction={
	pattern=north east lines
}] (axis cs:2.6,0) rectangle (axis cs:2.8,21.5167586179422);
\draw[draw=white,fill=color2,postaction={
	pattern=north east lines
}] (axis cs:-0.2,0) rectangle (axis cs:0,89.9295701096287);
%\addlegendentry{MGESSA}

\draw[draw=white,fill=color2,postaction={
	pattern=north east lines
}] (axis cs:0.8,0) rectangle (axis cs:1,83.4415320346072);
\draw[draw=white,fill=color2,postaction={
	pattern=north east lines
}] (axis cs:1.8,0) rectangle (axis cs:2,77.1956756133757);
\draw[draw=white,fill=color2,postaction={
	pattern=north east lines
}] (axis cs:2.8,0) rectangle (axis cs:3,66.770575151683);
\draw[draw=white,fill=color1] (axis cs:2.77555756156289e-17,0) rectangle (axis cs:0.2,21.0115143756204);
%\addlegendimage{ybar,ybar legend,draw=white,fill=color1}
%\addlegendentry{MSSSA}

\draw[draw=white,fill=color1] (axis cs:1,0) rectangle (axis cs:1.2,20.4117048441931);
\draw[draw=white,fill=color1] (axis cs:2,0) rectangle (axis cs:2.2,19.4851404120188);
\draw[draw=white,fill=color1] (axis cs:3,0) rectangle (axis cs:3.2,18.7276327299676);
\draw[draw=white,fill=color2] (axis cs:0.2,0) rectangle (axis cs:0.4,241.015680597302);
%\addlegendimage{ybar,ybar legend,draw=white,fill=color2}
%\addlegendentry{MSESSA}

\draw[draw=white,fill=color2] (axis cs:1.2,0) rectangle (axis cs:1.4,181.438988060673);
\draw[draw=white,fill=color2] (axis cs:2.2,0) rectangle (axis cs:2.4,135.912294873842);
\draw[draw=white,fill=color2] (axis cs:3.2,0) rectangle (axis cs:3.4,95.0181929071657);
\addplot [line width=1.08pt, darkslategray66, forget plot]
table {%
-0.3 22.9314155416671
-0.3 23.721723081041
};
\addplot [line width=1.08pt, darkslategray66, forget plot]
table {%
0.7 22.5762178717575
0.7 23.2996593684328
};
\addplot [line width=1.08pt, darkslategray66, forget plot]
table {%
1.7 22.1632675206417
1.7 22.8614572735707
};
\addplot [line width=1.08pt, darkslategray66, forget plot]
table {%
2.7 21.159774492592
2.7 21.8912710621034
};
\addplot [line width=1.08pt, darkslategray66, forget plot]
table {%
-0.1 84.7875744578684
-0.1 95.5225471873788
};
\addplot [line width=1.08pt, darkslategray66, forget plot]
table {%
0.9 78.7516165105326
0.9 87.9070686675198
};
\addplot [line width=1.08pt, darkslategray66, forget plot]
table {%
1.9 73.3612552653697
1.9 81.0948573737811
};
\addplot [line width=1.08pt, darkslategray66, forget plot]
table {%
2.9 63.869691823705
2.9 70.2021153202052
};
\addplot [line width=1.08pt, darkslategray66, forget plot]
table {%
0.1 20.7033459008698
0.1 21.3245504723652
};
\addplot [line width=1.08pt, darkslategray66, forget plot]
table {%
1.1 20.0625213503606
1.1 20.7333902346238
};
\addplot [line width=1.08pt, darkslategray66, forget plot]
table {%
2.1 19.1769912680995
2.1 19.783420248543
};
\addplot [line width=1.08pt, darkslategray66, forget plot]
table {%
3.1 18.4387888216452
3.1 19.033185467714
};
\addplot [line width=1.08pt, darkslategray66, forget plot]
table {%
0.3 236.946885986154
0.3 245.094098948577
};
\addplot [line width=1.08pt, darkslategray66, forget plot]
table {%
1.3 178.21747587413
1.3 184.665296024853
};
\addplot [line width=1.08pt, darkslategray66, forget plot]
table {%
2.3 133.794117395051
2.3 137.897488390495
};
\addplot [line width=1.08pt, darkslategray66, forget plot]
table {%
3.3 93.6195666977117
3.3 96.5409416119229
};
\end{axis}

\end{tikzpicture}
	}
	\caption{SNR ($\gamma$) and throughput ($\Bar{C}$) of the selected \glspl{ue} vs. $h$ and the slot pattern.}
	\label{fig:sched}
\end{figure*}

\begin{figure*}[t!]
	\begin{subfigure}[b]{\linewidth}
		\centering
		\includegraphics[width=0.99\columnwidth]{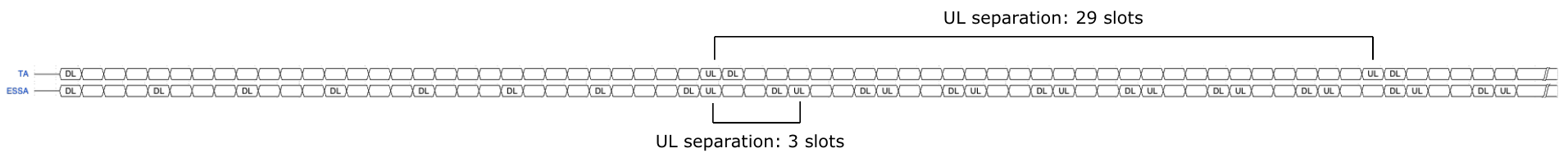}
  \caption{DSU slot pattern (where only one slot is allocated per DL transmission).}
  %\vspace{0.33cm}
	\end{subfigure}
 \begin{subfigure}[b]{\linewidth}
		\centering
		\includegraphics[width=0.99\columnwidth]{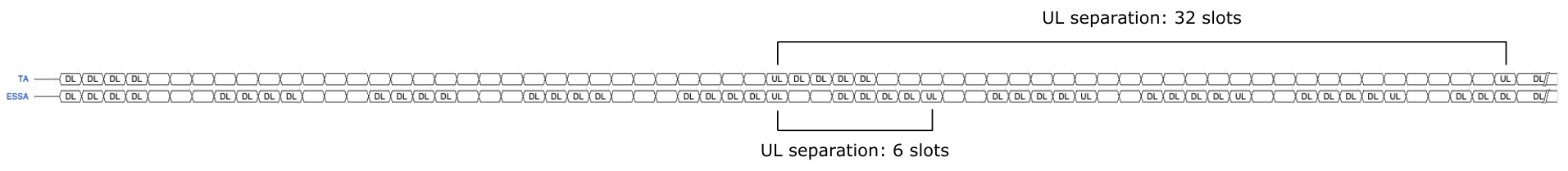}
  \caption{4DSU slot pattern (where four slots are scheduled
for each DL transmission).}
	\end{subfigure}
 \caption{Example of a slot pattern for a direct-to-satellite connection at $h=500$ km and $\alpha_{min}=80^\circ$. We compare TA vs. ESSA.\vspace{-0.5cm}}
 \label{fig:dsu}
\end{figure*}

\subsubsection*{Satellite altitude}
First, we explore the impact of $h$.
In \cref{fig:snr} we observe that the path loss increases with $h$, so the SNR decreases as $h$ increases.
For instance, at $h=300$~km, the median \gls{snr} with MG scheduling is approximately 29~dB, vs. 23~dB with MS scheduling.
This gap is because MS is based on the propagation delay alone, with no considerations related to the channel quality.

In~\cref{fig:capacity} we plot the average capacity of the selected $N_s$ UEs via MS or MG scheduling vs. $h$, considering both TA and ESSA for slot allocation.
We see that MS-ESSA outperforms MG-ESSA. 
In fact, MS is specifically optimized for ESSA. 
In this scenario, additional DL slots can be allocated if $2\tau_m \geq t_{\rm UL}$: MS prioritizes UEs with the lowest differential propagation delays based on~\cref{eq:S_MS} in an attempt to reduce $\tau_m$ and schedule more transmissions, which permits to improve the network capacity. 
%In contrast, in MG, although more bits can be transmitted per slot, the lower number of slots allocated results in an overall increase in throughput.
At the same time, MG-TA outperforms MS-TA since MG is explicitly designed to maximize the number of bits per slot and so the capacity.
Between the two best configurations, {MS-ESSA} achieves more than three times the capacity of {MG-TA}, despite the lower \gls{snr} of MS compared to MG.
This is because MG-TA is affected by the limitations of TA in terms of overhead and channel usage. In turn, MS-ESSA can scheduled additional slots during guard periods, even if fewer bits per slot are transmitted, which eventually improves resource utilization.

\subsubsection*{Slot pattern}
We now analyze the effect of the slot pattern.
In particular, we consider the case in which a DL transmission requires more than one DL slot to be completed, e.g., to satisfy more aggressive traffic requests. The notation XDSU is used to indicate that X consecutive slots are consumed for each DL transmission.
As expected, increasing the number of slots dedicated to DL transmissions improves the DL capacity. 
Considering {MG-TA}, in~\cref{fig:capacity_slot_pattern} we see that $\Bar{C}_{\rm DL}$ increases from approximately 25~Mbps with DSU (where only one slot is allocated per \gls{dl} transmission) to nearly 130~Mbps with 6DSU (where six slots are scheduled for each \gls{dl} transmission).
 %Similarly, MS-ESSA also benefits from an extended slot pattern, and achieves a \gls{dl} throughput of approximately 570 Mbps with 6DSU.

These results also confirm that {MS-ESSA} outperforms {MG-TA}. 
%In fact, ESSA permits to allocate multiple consecutive slots within the guard period, provided that the condition in Lemma 1 is satisfied, thereby mitigating the impact of overhead within the same timeframe.
As shown in~\cref{fig:capacity_slot_pattern}, the average \gls{dl} capacity for TA is up to 130~Mbps with 6DSU, vs. around 570~Mbps for ESSA in the same configuration, i.e., more than 4 times higher.
However, in {MS-ESSA} the use of more \gls{dl} slots per transmission comes at a significant cost in terms of \gls{ul} capacity.
In~\cref{fig:tx_slot_pattern} we see that, while in MG-TA $\Bar{C}_{\rm UL}$ is almost constant as the number of DL slots increases, in MS-ESSA it drops from around 240~Mbps with DSU to only 95~Mbps with 6DSU.
In fact, in MG-TA the frame structure is sparse due to guard periods, so new DL slots can be allocated by simply shifting UL slots accordingly.
In contrast, in MS-ESSA transmissions are scheduled also during guard periods, provided that the condition in Lemma 1 is satisfied, leaving limited flexibility for allocating new DL slots, at the expense of UL slots.
In~\cref{fig:dsu} we compare DSU vs. 4DSU for a direct-to-satellite connection at $h=500$ km and $\alpha_{min}=80^\circ$. In this example, the guard period is constant and equal to 29 slots in TA, so UL slots are separated by $29$ slots in DSU, vs. $29+(4-1)=32$ in 4DSU. The UL overhead therefore increases by as little as 10\% in 4DSU.
In ESSA, instead, UL slots are separated by $3$ slots in DSU, vs. $6$ in 4DSU, so the UL overhead is around 50\%. 

It is clear that ESSA introduces some bias in the scheduling of DL and UL transmissions, where the latter are sacrificed in favor of the former.
As part of our future work, we will improve the design of ESSA to incorporate additional fairness in the allocation of UL and DL slots.
%hile in MG-TA the UL throughput is almost constant as the number of DL slots increases, 

\section{Conclusions and Future Works}
\label{sec:conclusions} 
The design of advanced \gls{tdd} schemes can facilitate the integration of \glspl{tn} and \glspl{ntn}, ultimately enabling seamless switching between satellite and terrestrial connectivity in uncovered areas or in case of emergency.
In this work we proposed an enhanced slot allocation mechanism for \gls{tdd} \gls{ntn} called ESSA. The goal is to mitigate the impact of the long propagation delays experienced in satellite networks by allocating additional transmission slots that do not cause interference during guard periods. Moreover, we proposed two scheduling methods to select the optimal UEs that should transmit in TDD to maximize the network capacity.
We showed that ESSA can reduce the overhead in DL compared to a baseline solution that simply implements TA to avoid interference, despite some degradation of the UL~capacity.

In this paper, the assumption was to design NTNs that operate in TDD like TNs, as promoted in 3GPP NTN Release 17. In our future research, we will also focus on the seamless integration between NTN FDD and TN TDD systems, so as to reduce the complexity of NTN transmissions.

\section*{Acknowledgments}
This work was supported by the European Commission through the European Union's Horizon Europe Research and Innovation Programme under the Marie Skłodowska-Curie-SE, Grant Agreement No. 101129618, UNITE.
This research was also partially supported by the European Union under the Italian National Recovery and Resilience Plan (NRRP) of NextGenerationEU, 
partnership on ``Telecommunications of the Future'' (PE0000001 - program “RESTART”).

\vspace{1cm}
\bibliographystyle{IEEEtran}
\bibliography{bibliography.bib}

\end{document}